\documentclass[sigconf,balance=false]{acmart}

\usepackage{graphicx}
\usepackage{subcaption}
\usepackage{tikz}
\usetikzlibrary{positioning}
\usepackage{dsfont}
\usepackage{mathtools}
\usepackage{algorithm}
\usepackage[noend]{algpseudocode}
\usepackage{array}
\usepackage{pgfplots}
\usepackage{pgfplotstable}
\usepgfplotslibrary{fillbetween}
\pgfplotsset{compat=1.18}
\usepackage{newunicodechar}
\usepackage{stfloats}
\newunicodechar{ε}{\ensuremath{\epsilon}}
\newunicodechar{τ}{\ensuremath{\tau}}
\newunicodechar{μ}{\ensuremath{\mu}}
\newunicodechar{σ}{\ensuremath{\sigma}}
\newunicodechar{δ}{\ensuremath{\delta}}
\newunicodechar{Δ}{\ensuremath{\Delta}}
\newunicodechar{−}{-}  
\usepackage{xcolor}
\definecolor{myBlue}{RGB}{33,113,181}
\definecolor{myRed}{RGB}{203,24,29}
\definecolor{myGreen}{RGB}{35,132,67}
\definecolor{myOrange}{RGB}{217,95,14}
\definecolor{myPurple}{RGB}{117,107,177}

\pgfplotscreateplotcyclelist{mycolors}{
  {myBlue},
  {myRed},
  {myGreen},
  {myOrange},
  {myPurple},
}

\pgfplotscreateplotcyclelist{mystyles}{
  {solid, mark=o, mark size=1.5pt, mark options={solid, fill=white}},
  {dashed, mark=none},
}

\algrenewcommand\algorithmicrequire{\textbf{Input:}}
\algrenewcommand\algorithmicensure{\textbf{Output:}}

\DeclareMathOperator*{\argmax}{argmax}
\DeclareMathOperator{\diag}{\mathsf{diag}}
\DeclareMathOperator{\triu}{\mathsf{triu}}

\DeclareMathOperator{\cmp}{cmp}
\DeclareMathOperator{\ind}{ind}
\DeclarePairedDelimiter{\ceil}{\lceil}{\rceil}
\DeclarePairedDelimiter{\floor}{\lfloor}{\rfloor}
\newcommand\norm[1]{\left\lVert#1\right\rVert}

\newcommand{\mytilde}{\raisebox{0.5ex}{\texttildelow}}
\newcommand{\repository}{\url{https://github.com/FedericoMazzone/secure-change-point-detection}}

\newcommand{\enc}{\mathsf{Enc}}

\newcommand{\sumR}{\mathsf{SumR}}
\newcommand{\sumC}{\mathsf{SumC}}
\newcommand{\replR}{\mathsf{ReplR}}
\newcommand{\replC}{\mathsf{ReplC}}
\newcommand{\transR}{\mathsf{TransR}}
\newcommand{\transC}{\mathsf{TransC}}
\newcommand{\rankAlg}{\mathsf{Rank}}
\newcommand{\argmaxAlg}{\mathsf{Argmax}}
\newcommand{\meanAlg}{\mathsf{Mean}}
\newcommand{\varAlg}{\mathsf{Variance}}
\newcommand{\turnRatesAlg}{\mathsf{TurningRates}}
\newcommand{\cusumAlg}{\mathsf{CUSUM}}
\newcommand{\cpdAlg}{\mathsf{CPD}}

\newcommand{\vectorLength}{N}
\renewcommand{\vector}{x}
\newcommand{\encVector}{X}
\newcommand{\encRanking}{R}
\newcommand{\encComparison}{C}
\newcommand{\encArgmax}{A}

\newcommand{\tabcolsepValue}{1.3pt}
\newcommand{\arraystretchValue}{1.5}
\newcolumntype{P}[1]{>{\centering\arraybackslash}p{#1}}

\setcopyright{none}
\renewcommand\footnotetextcopyrightpermission[1]{}

\title[Secure Change-Point Detection]{Secure Change-Point Detection for Time Series under Homomorphic Encryption}

\author{Federico Mazzone}
\authornote{The two authors made equal contributions and share first authorship.}
\orcid{0000-0002-6316-6409}
\affiliation{%
  \institution{University of Twente}
  \city{Enschede}
  \state{}
  \country{Netherlands}
}
\email{f.mazzone@utwente.nl}

\author{Giorgio Micali}
\authornotemark[1]  
\orcid{0009-0002-8752-890X}
\affiliation{%
  \institution{University of Twente}
  \city{Enschede}
  \state{}
  \country{Netherlands}
}
\email{g.micali@utwente.nl}

\author{Massimiliano Pronesti}
\affiliation{%
  \institution{IBM Research Europe}
  \city{Dublin}
  \state{}
  \country{Ireland}
}
\email{massimiliano.pronesti@ibm.com}

\begin{document}

\begin{abstract}
We introduce the first method for change-point detection on encrypted time series.
Our approach employs the CKKS homomorphic encryption scheme to detect shifts in statistical properties (e.g., mean, variance, frequency) without ever decrypting the data.
Unlike solutions based on differential privacy, which degrade accuracy through noise injection, our solution preserves utility comparable to plaintext baselines.
We assess its performance through experiments on both synthetic datasets and real-world time series from healthcare and network monitoring.
Notably, our approach can process one million points within 3 minutes.
\end{abstract}


\maketitle
\pagestyle{plain}

\section{Introduction}

The analysis of time series data often requires detecting structural changes, known as \textit{change points}, which indicate significant shifts in the underlying data distribution.
The change-point detection (CPD) problem consists of identifying time indices at which such shifts occur. This problem has been studied for over seventy years, dating back to the pioneering work of Page~\cite{page1955test}.
Since then, CPD has become a central topic of research across many areas of quantitative and applied science: in econometrics and finance~\cite{lavielle2007adaptive,bai1998estimating,DAVIES20123623}, network monitoring~\cite{levy2009detection,lung2012distributed}, signal processing~\cite{Blythe2012}, genetics~\cite{braun2000}, climatology~\cite{verbesselt2010detecting,reeves2007review}, quantum optics~\cite{SCHMIDT20129}, speech recognition~\cite{Chowdhury2011}, MRI analysis~\cite{aston2012evaluating}, and healthcare~\cite{Staudacher2005}.
Across these domains, a change in the data distribution may indicate the onset of a disease or a shift in a patient's condition, a market crash or policy intervention, a sudden variation in climatic patterns, or an anomaly in network traffic. 
Accurately detecting such changes is essential both for retrospective analysis and for real-time decision-making.

Time-series data are often collected by sensors or client devices and transmitted to external servers or data curators for processing.
In many practical scenarios, on-device (local) analysis is either infeasible or undesirable.
For instance, IoT monitoring devices typically lack permanent storage and are mainly designed to collect, optionally encrypt, and forward data rather than to run specialized analytics such as CPD.
Even when feasible, embedding task-specific logic would reduce their flexibility, as analysts may not know in advance which time intervals to inspect or which preprocessing steps to apply.
Outsourcing the raw data first and performing the analysis off-device offers greater versatility and allows multiple tasks, such as CPD, anomaly detection, or other retrospective analyses, to be executed later.
Moreover, many data providers already rely on cloud-based platforms (e.g., Azure Health Data Services or Google Healthcare API) for secure storage, regulatory compliance, or cross-institutional collaboration, making outsourced CPD a natural fit.
Such a setting is also common when data originate from multiple parties and a third-party curator is required to coordinate tasks such as timeline alignment, normalization, or data fusion, which cannot be performed locally.

However, outsourcing raises significant privacy concerns, as the transmitted time-series data are often highly sensitive.
For example, medical time series such as ECGs or EEGs contain personal health information; network traffic traces may reveal private user behavior or communication patterns; and industrial monitoring data can expose proprietary system activity.
Previous work~\cite{NEURIPS2018_f19ec2b8,JMLR:v22:19-770} addressed privacy-preserving CPD exclusively in the \textit{central} privacy model~\cite{dwork2006calibrating}, which assumes the data curator to be trusted and focuses on providing privacy at the moment the computed change-point index is released to an external party.
To do so, they apply differential privacy (DP) by injecting noise into the CPD mechanism, so that only a noisy version of the result is released.
This limits what a potential adversary can infer about the underlying time series, at the cost of reduced accuracy.

In many settings, however, the data provider does not trust the curator, placing the problem in the \textit{local} privacy model.
In this setting, the client must protect its data before sending it out.
To date, no work has addressed privacy-preserving CPD in the local model, and this work aims to close this gap.
In particular, \textbf{our goal is to enable an honest-but-curious curator} (or computing server) \textbf{to perform CPD on the client's data without compromising the privacy of the underlying time series}.
One option would be to use local DP, where the client perturbs the raw time-series values with random noise before sending it to the curator.
As we will show, injecting noise directly at the data level severely degrades the accuracy of the CPD task, making this approach impractical in many realistic scenarios (over 30\% relative error in the resulting change-point index for a 10,000-point series with $\varepsilon = 25$).
To avoid this limitation, we instead employ fully homomorphic encryption (FHE), specifically the CKKS scheme~\cite{cheon2017homomorphic}.
In our solution, the client encrypts its entire time series and sends only ciphertexts to the curator.
Thanks to the homomorphic properties of CKKS, the server can then execute the full CPD pipeline directly over encrypted data, without ever accessing the underlying values, and return the encrypted result to the client.
In addition to avoiding the accuracy issues of local DP, this approach also offers a stronger privacy guarantee, as the curator never sees the data, not even in noisy~form.

These improved accuracy and privacy guarantees come at the cost of a large computational overhead due to the homomorphic operations.
Nonetheless, our choice of FHE over DP is still motivated by the fact that, in many applications, accuracy requirements tend to outweigh runtime constraints.
In clinical monitoring, for instance, detecting a physiological regime change (e.g., heart-rate) within a margin of a few seconds can distinguish between early-stage and late-stage patient deterioration.
By contrast, runtime can tolerate delays, since such analyses are commonly executed offline or asynchronously with respect to data collection.
Similarly, in network-traffic monitoring, accurately spotting sudden changes in packet-rate or latency is important to distinguish normal fluctuations from the onset of a DDoS attack or a routing fault.
Missing such a change can delay incident detection or trigger false alarms, while some computational delay is typically acceptable since analyses often run on buffered or batched traffic.

\paragraph{Our Approach and Contributions}
Prior works~\cite{NEURIPS2018_f19ec2b8,JMLR:v22:19-770} build their approaches on classical detection techniques based on the log-likelihood ratio principle.
However, these techniques rely on strong statistical assumptions that rarely hold in practice:
\begin{enumerate}
    \item independence among the data points of the time series, and
    \item prior knowledge of the data distribution before and after the change point.
\end{enumerate}
To avoid the need for such assumptions, we base our solution on the detection strategy proposed by Bandt and Pompe~\cite{bandt2020order} and Betken et al.~\cite{betken2025ordinal}.
This strategy consists of first transforming the time series into a sequence of ordinal-pattern histograms (i.e., permutations encoding the relative ordering of observations), reducing the problem to an easier mean-change detection task, and then applying standard Cumulative Sum (CUSUM) statistics to solve it on the resulting summaries.
Since ordinal patterns provide a non-parametric description of the time series, they are naturally robust to noisy observations and outliers, and capture temporal structure without assuming independence or specific data distributions.
As a result, they are well suited for detecting changes in temporal dependence (e.g., frequency-related changes), but are inherently insensitive to changes in marginal properties such as the mean or variance, for which we implement different transformations.

The main challenge in implementing this CPD strategy under FHE lies in performing the required computations efficiently, given both the high cost of homomorphic operations and the large size of typical time-series data. 
Our framework builds on recent advances in encrypted comparison techniques~\cite{mazzone2025efficient}, which enable an efficient homomorphic evaluation of ranking and order-statistics operations.
Our approach proceeds in two steps.
First, we homomorphically transform the input time series into a sequence of block-level summaries, implementing different transformations for detecting changes in mean, variance, and frequency.
Second, we implement the CUSUM statistic under encryption and apply it to these block-level sequences to localize the change point.
We evaluate our solution on multiple datasets, including synthetic data from standard toy models and real-world data from healthcare and network monitoring.
Notably, our method can process a time series with 1,000,000 points, detecting mean changes in under two minutes and frequency and variance changes in under three minutes, all while preserving plaintext-level accuracy.

The main contributions of this paper are as follows.
\begin{itemize}
    \item We present the first solution for secure outsourced CPD on dependent time series, achieving plaintext-level accuracy and practical runtime for detecting mean, variance, and frequency shifts under CKKS.
     \item We introduce new cryptographic building blocks for CKKS, including an efficient method for partial sums and an improved implementation of the argmax operation.
    \item We provide a thorough experimental evaluation of our solution on both synthetic and real-world datasets.
\end{itemize}

\section{Background}

\subsection{System and Threat Model}

The most general system model for CPD computation involves the following actors:
\begin{itemize}
    \item one or more data providers,
    \item a computing server or data curator, and
    \item a result recipient.
\end{itemize}
The ideal functionality is that the data providers send their time-series data to the computing server, which computes the CPD index and delivers the result to the recipient.
In this work, we focus on the outsourced-computation setting, hence, with a single data provider, who also coincides with the result recipient.
The data provider encrypts its time series locally and sends the resulting ciphertexts to the server, which performs the entire CPD algorithm homomorphically and returns the encrypted output.
Only the data provider, who holds the secret key, can decrypt the result.

\paragraph{Threat Model}
We consider the computing server to be honest-but-curious, i.e., it follows the protocol but may attempt to extract information about the data provider's time series from any intermediate computation it observes.
In particular, the server is assumed to have full access to all ciphertexts and to all homomorphic computations performed during CPD, and its goal is to infer statistical properties, patterns, or individual data values from them.
Security against such an adversary follows directly from the semantic security of the encryption scheme.
Note that our threat model does not cover a disclosure of the decrypted result to a third-party recipient.
In that case, central-DP mechanisms can be employed to protect the released index.

\subsection{Time Series and the CPD Problem}
A time series \((X_t)_{t \geq 1}\) is a sequence of random variables indexed by discrete time \(t\). 
A standard approach in time series analysis is to model the data as a realization of a \emph{stationary process}, where the statistical properties such as mean, variance, and autocorrelation do not change over time~\cite{brockwelldavis}. 
Formally, stationarity means that the joint distribution of \((X_{t_1}, \ldots, X_{t_k})\) is the same as that of \((X_{t_1+h}, \ldots, X_{t_k+h})\) for all time shifts \(h \in \mathbb{Z}\) and all index tuples \(t_1, \ldots, t_k\).
This assumption enables meaningful inference, as observations from different time intervals can be treated as samples from the same underlying distribution.
However, many real-world time series exhibit stationary behavior only over certain intervals and undergo structural changes at unknown times.
In such cases, the global stationarity assumption fails, and detecting these change-points becomes essential for proper modeling and interpretation.
Sometimes changes can be recognized visually (see Figure~\ref{fig:change-point-examples} for some examples), but in many cases they are too subtle to detect by eye, or appear in contexts where manual inspection is not feasible.
To automate this task, a variety of methods have been proposed, most of which rely on a hypothesis-testing formulation.

\begin{figure}
\centering
\newcommand{\subfigheight}{60pt}
\begin{subfigure}[b]{\columnwidth}
    \centering
    \includegraphics[width=\linewidth,height=\subfigheight]{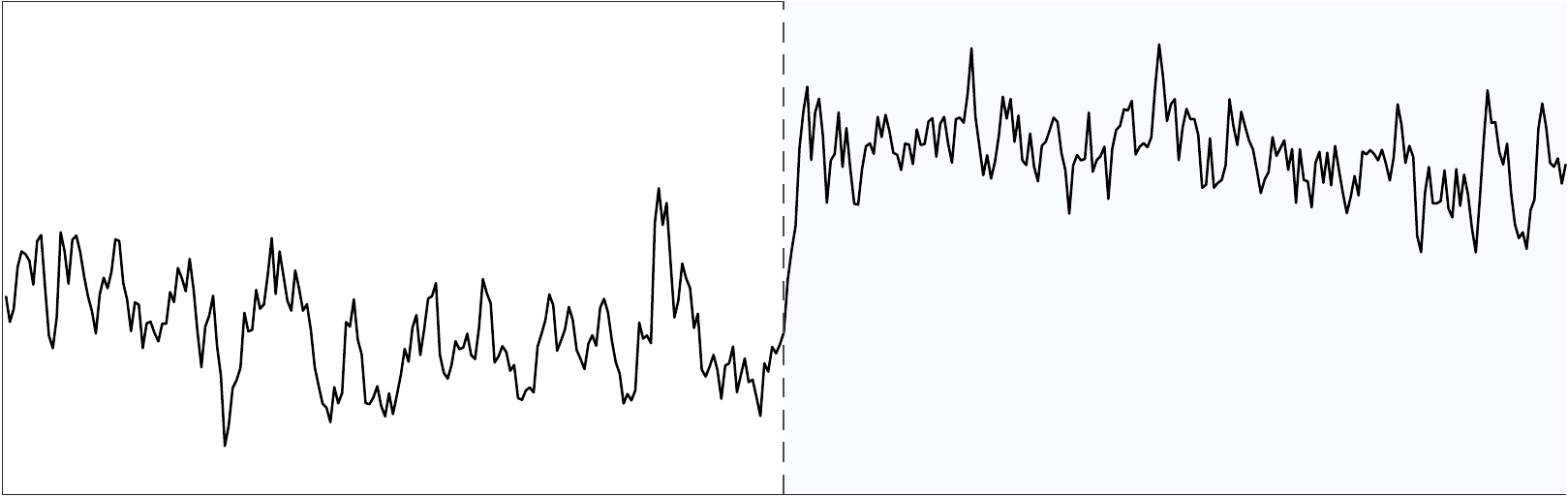}
    \caption{Change in mean}
\end{subfigure}
\vspace{0pt}

\begin{subfigure}[b]{\columnwidth}
    \centering
    \includegraphics[width=\linewidth,height=\subfigheight]{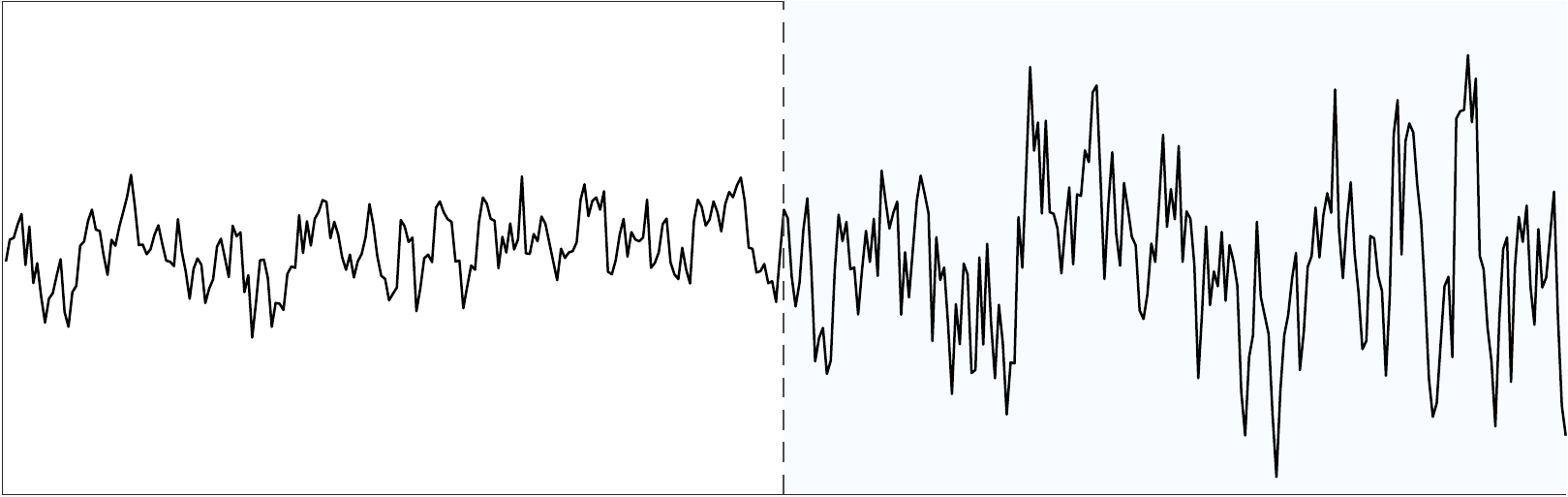}
    \caption{Change in variance}
\end{subfigure}
\vspace{0pt}

\begin{subfigure}[b]{\columnwidth}
    \centering
    \includegraphics[width=\linewidth,height=\subfigheight]{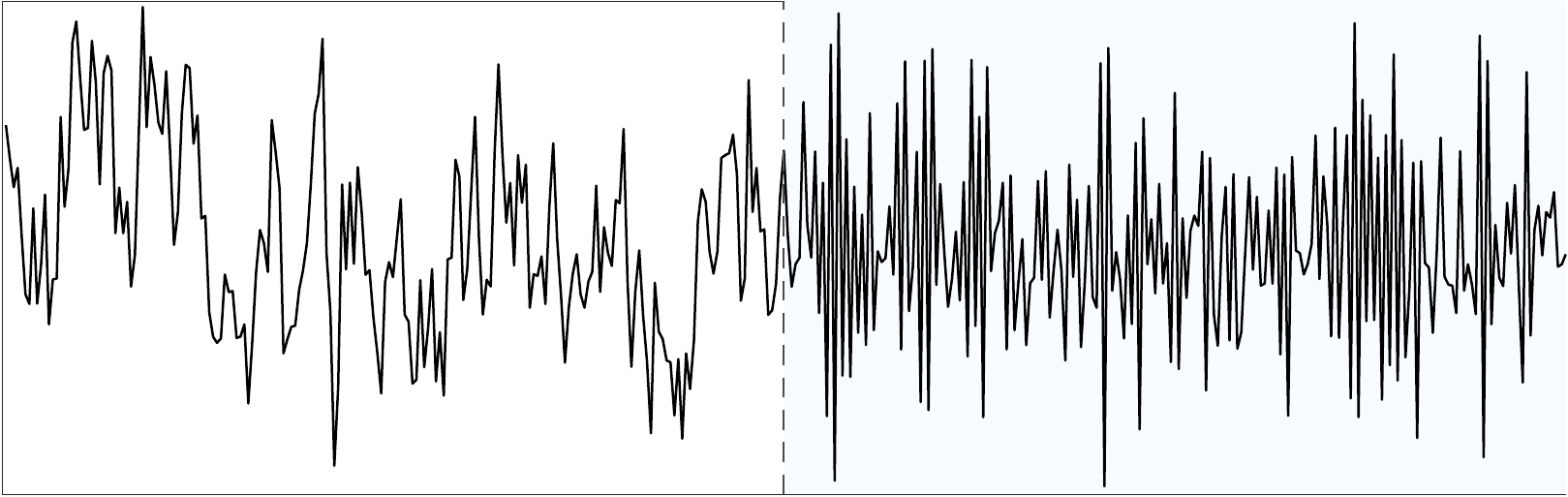}
    \caption{Change in frequency}
\end{subfigure}

\caption{Three common types of change points in time series.}
\label{fig:change-point-examples}
\Description{}
\end{figure}

Related works in the privacy-preserving context~\cite{NEURIPS2018_f19ec2b8,JMLR:v22:19-770} formalize the change-point detection problem by selecting a transformation~\(\psi\) and testing whether its expectation remains constant over~time:
\begin{align*}
\mathcal{H}_0 &: \mathbb{E}[\psi(X_t)] = \mu, \ \forall t \\
\mathcal{H}_1 &: \exists\, \tau\text{ such that } \mathbb{E}[\psi(X_{\tau})] \neq \mathbb{E}[\psi(X_{\tau+1})]\;.
\end{align*}
The choice of~$\psi$ determines the type of change being tested.
For instance, $\psi(x) = x$ corresponds to changes in mean (location), while $\psi(x) = x^2$ corresponds to changes in variance (volatility).
In fact, a classic test for detecting mean shifts is the following:
\begin{align*}
\mathcal{H}_0\!: &\; X_t = \mu + \varepsilon_t,\quad \varepsilon_t \text{ i.i.d.},\quad\forall t \\
\mathcal{H}_1\!: &\; X_t = 
\begin{cases}
\mu_1 + \varepsilon_t, & t = 1, \dots, \tau, \\
\mu_2 + \varepsilon_t, & t = \tau + 1, \dots, n
\end{cases}
\quad \text{with } \mu_1 \neq \mu_2.
\end{align*}
where $n$ is the length of the series.

While the $\psi$-formulation is convenient, it only captures specific aspects of the distribution and ignores the temporal dependence usually present in time-series data, making it inadequate in realistic settings.
Hence, in this work we focus on the more general formulation:
$$
\begin{aligned}
\mathcal{H}_0 &: (X_1, \ldots, X_n) \text{ is a stationary process} \\
\mathcal{H}_1 &: \exists\, \tau \in \{1, \dots, n-1\} \text{ such that }\\ & (X_1, \ldots, X_{\tau}) \overset{\mathcal{D}}{\neq} (X_{\tau+1}, \ldots, X_n)\;.
\end{aligned}
$$
Our goal is to privately estimate the location of the change point~$\tau$ under the alternative hypothesis.
To this end, we build upon an algorithm that transforms the time series into a symbolic representation known as \emph{ordinal patterns}.

We refer the reader to review papers by Gupta et al. \cite{GUPTA2024123342} and Jandhyala et al. \cite{jandhyala2013inference}, and the monographs by Basseville and Nikiforov~\cite{basseville1993abrupt} and Csörgő and Horváth~\cite{CsorgoHorvath1997} for comprehensive treatments of change-point theory and methods.

\subsection{Ordinal Patterns Representation}
\label{sec:background:ordpatt}

To estimate the change-point location under the general formulation, we adopt the approach introduced by Bandt and Pompe~\cite{Bandt-Pompe}, who proposed to encode the dynamics of a time series as a stream of symbols.
These symbols correspond to ordinal patterns, each describing the relative order of values within a local window of the series.
Formally, let $\mathcal{S}_r$ denote the set of permutations of $\{0,1,\ldots,r\}$.
Given a real-valued time series $(X_1,\ldots,X_n)$ and an embedding order $r \geq 1$, we define the ordinal pattern map
\begin{equation}\label{eq:def_op}
\Pi: \mathbb{R}^{r+1} \rightarrow \mathcal{S}_r, \quad (X_t, \ldots, X_{t+r}) \mapsto \pi,
\end{equation}
where $\pi=(\pi_0,\ldots,\pi_r)$ is the permutation that sorts the values $(X_t,\ldots,X_{t+r})$, i.e., $\pi_j$ denotes the rank of $X_{t+j}$ within the window.
Sliding the window along the series yields a symbolic sequence of ordinal patterns. 

\paragraph{Example}
Consider the time series
\[
\{x_1,x_2,\ldots\}=\{4.2,\; 3.1,\; 5.0,\; 6.3,\; 2.9,\; 7.1,\; 1.8,\; 3.7,\; \ldots\}
\]
with embedding order $r=2$, so each ordinal pattern is extracted from a window of three consecutive values.
The first window $(4.2,\; 3.1,\; 5.0)$ is sorted as $(3.1,\; 4.2,\; 5.0)$, yielding the ranks $(2,1,3)$.
The next window $(3.1,\; 5.0,\; 6.3)$ gives $(1,2,3)$, and so on.
Thus, the time series is transformed into a sequence of rank tuples:
\[
\{(2,1,3),\; (1,2,3),\; (2,3,1),\; (2,1,3),\; \ldots \}.
\]
Figure~\ref{fig:op} illustrates the six possible ordinal patterns of order $r=2$.

\begin{figure}
\resizebox{\columnwidth}{!}{
\centering
\begin{tikzpicture}[x=1pt,y=1pt]
\path[use as bounding box,fill=black,fill opacity=0.00] (0,0) rectangle (361.35, 72.27);
\begin{scope}
\path[draw=black,line width= 0.6pt] ( 10.68, 10.19) --
	( 34.93, 29.58) --
	( 59.17, 48.97);
\path[draw=black,line width= 0.4pt,line join=round,line cap=round,fill=black] ( 10.68, 10.19) circle (  2.50);
\path[draw=black,line width= 0.4pt,line join=round,line cap=round,fill=black] ( 34.93, 29.58) circle (  2.50);
\path[draw=black,line width= 0.4pt,line join=round,line cap=round,fill=black] ( 59.17, 48.97) circle (  2.50);
\end{scope}

\begin{scope}
\path[draw=black,line width= 0.6pt,line join=round] ( 69.53, 10.19) --
	( 93.77, 48.97) --
	(118.02, 29.58);
\path[draw=black,line width= 0.4pt,line join=round,line cap=round,fill=black] ( 69.53, 10.19) circle (  2.50);
\path[draw=black,line width= 0.4pt,line join=round,line cap=round,fill=black] ( 93.77, 48.97) circle (  2.50);
\path[draw=black,line width= 0.4pt,line join=round,line cap=round,fill=black] (118.02, 29.58) circle (  2.50);
\end{scope}

\begin{scope}
\path[draw=black,line width= 0.6pt,line join=round] (128.37, 48.97) --
	(152.62, 10.19) --
	(176.87, 29.58);
\path[draw=black,line width= 0.4pt,line join=round,line cap=round,fill=black] (128.37, 48.97) circle (  2.50);
\path[draw=black,line width= 0.4pt,line join=round,line cap=round,fill=black] (152.62, 10.19) circle (  2.50);
\path[draw=black,line width= 0.4pt,line join=round,line cap=round,fill=black] (176.87, 29.58) circle (  2.50);
\end{scope}

\begin{scope}
\path[draw=black,line width= 0.6pt,line join=round] (187.22, 48.97) --
	(211.47, 29.58) --
	(235.72, 10.19);
\path[draw=black,line width= 0.4pt,line join=round,line cap=round,fill=black] (187.22, 48.97) circle (  2.50);
\path[draw=black,line width= 0.4pt,line join=round,line cap=round,fill=black] (211.47, 29.58) circle (  2.50);
\path[draw=black,line width= 0.4pt,line join=round,line cap=round,fill=black] (235.72, 10.19) circle (  2.50);
\end{scope}

\begin{scope}
\path[draw=black,line width= 0.6pt,line join=round] (246.07, 29.58) --
	(270.32, 48.97) --
	(294.57, 10.19);
\path[draw=black,line width= 0.4pt,line join=round,line cap=round,fill=black] (246.07, 29.58) circle (  2.50);
\path[draw=black,line width= 0.4pt,line join=round,line cap=round,fill=black] (270.32, 48.97) circle (  2.50);
\path[draw=black,line width= 0.4pt,line join=round,line cap=round,fill=black] (294.57, 10.19) circle (  2.50);
\end{scope}

\begin{scope}
\path[draw=black,line width= 0.6pt,line join=round] (304.92, 29.58) --
	(329.17, 10.19) --
	(353.42, 48.97);
\path[draw=black,line width= 0.4pt,line join=round,line cap=round,fill=black] (304.92, 29.58) circle (  2.50);
\path[draw=black,line width= 0.4pt,line join=round,line cap=round,fill=black] (329.17, 10.19) circle (  2.50);
\path[draw=black,line width= 0.4pt,line join=round,line cap=round,fill=black] (353.42, 48.97) circle (  2.50);
\end{scope}

\definecolor{text}{gray}{0.10}
\begin{scope}
\node[text=text,anchor=base,inner sep=0pt, outer sep=0pt] at ( 34.93, 56.08) {(1, 2, 3)};
\end{scope}

\begin{scope}
\node[text=text,anchor=base,inner sep=0pt, outer sep=0pt] at ( 93.77, 56.08) {(1, 3, 2)};
\end{scope}

\begin{scope}
\node[text=text,anchor=base,inner sep=0pt, outer sep=0pt] at (152.62, 56.08) {(3, 1, 2)};
\end{scope}

\begin{scope}
\node[text=text,anchor=base,inner sep=0pt, outer sep=0pt] at (211.47, 56.08) {(3, 2, 1)};
\end{scope}

\begin{scope}
\node[text=text,anchor=base,inner sep=0pt, outer sep=0pt] at (270.32, 56.08) {(2, 3, 1)};
\end{scope}

\begin{scope}
\node[text=text,anchor=base,inner sep=0pt, outer sep=0pt] at (329.17, 56.08) {(2, 1, 3)};
\end{scope}
\end{tikzpicture}}
\caption{The six ordinal patterns of order $2$ (figure from~\cite{betken2025ordinal}).}
\label{fig:op}
\Description{}
\end{figure}
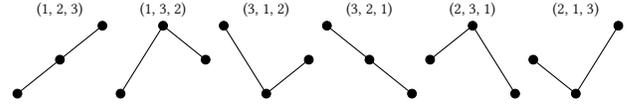

\medskip
This symbolic representation offers multiple advantages: it is invariant under strictly increasing transformations of the data (such as rescaling or nonlinear monotonic mappings), and it is robust to outliers, since the ranks depend only on the relative order of the values rather than their magnitudes.
Moreover, it induces a probability distribution over $\mathcal{S}_r$ that captures the dependence structure of the process.
For each $\pi \in \mathcal{S}_r$, we define the pattern probability
\begin{equation}\label{eq:pattern_prob}
p(\pi) := \mathbb{P}\!\left( \Pi(X_t, \ldots, X_{t+r}) = \pi \right),
\end{equation}
which can be estimated from a finite sample via the empirical frequency
\begin{equation}\label{ordinal_pattern_estimator}
\hat{p}_n(\pi) = \frac{1}{n} \sum_{t=1}^{n-r} \mathds{1} \!\left( \Pi(X_t, \ldots, X_{t+r}) = \pi \right).
\end{equation}
This distribution serves as a compact summary of the time series.
It can be shown that under stationarity, $\hat{p}_n$ remains constant over time, while a change in the underlying data-generating process induces a shift~\cite{betken2025ordinal}.
This way, CPD in real-valued time series reduces to detecting changes in discrete distributions over the \mbox{finite alphabet~$\mathcal{S}_r$.}

\subsection{Turning Rates and CPD Algorithm}
\label{sec:background:turningrate}

Several statistics can be derived from ordinal patterns to detect distributional changes. 
Among these, Bandt ~\cite{bandt2020order} introduced the \emph{turning rate}, which is based on ordinal patterns of order~2.
It is defined as the total probability mass assigned to the four up-down~patterns:
\[
\mathcal{T} = \{(1,3,2), (3,2,1), (2,3,1), (2,1,3)\}.
\]
The turning rate is then
\begin{equation}
q := \sum_{\gamma \in \mathcal{T}} p(\gamma),
\end{equation}
where $p(\gamma)$ is the pattern probability from~\eqref{eq:pattern_prob}.
Given observations $X_1,\ldots,X_n$, the series is partitioned into $n_b = \ceil{n / m}$ non-overlapping blocks of length $m$, and a separate estimate $\hat{q}_{j,m}$ is computed for each block. 
This produces a sequence of local statistics that reflects how the turning rate evolves over time and can therefore be analyzed for structural breaks.
Concretely, for each block we compute
\begin{align}
\label{turning_rate_block}
\hat{q}_{j,m} &= \frac{1}{m} \sum_{i=1}^{m-2} \sum_{\gamma \in \mathcal{T}} 
\mathds{1}\!\left( 
\Pi(X_{s_{i,j}}, X_{s_{i,j}+1}, X_{s_{i,j}+2}) = \gamma 
\right), \\
s_{i,j} &:= m(j-1) + i, \quad j=1,\ldots,n_b. \nonumber
\end{align}
This reformulates the original stationarity testing problem as a mean change detection problem on the sequence $(\hat{q}_{j,m})_{j=1}^{n_b}$.
Specifically, the hypotheses become
\begin{align}
\label{test2}
\mathcal{H}_0:\,& \mathbb{E}[\hat{q}_{j,m}] = \mu_1 \quad \forall j, \\
\mathcal{H}_1:\,& \exists \tau^\ast\in \{1,\dots,n_b\} \text{ s.t. } \nonumber \\
&\mathbb{E}[\hat{q}_{j,m}] = 
\begin{cases}
\mu_1, & j = 1, \ldots, \tau^\ast, \\
\mu_2, & j = \tau^\ast + 1, \ldots, n_b,
\end{cases} \quad \mu_1 \neq \mu_2. \nonumber
\end{align}
To evaluate this mean-change hypothesis, Betken et al.~\cite{betken2025ordinal} use the classical CUSUM statistic applied to $(\hat{q}_{j,m})_j$, which reaches its maximum absolute value at the mean shift location:
\begin{equation}
\label{cusum}
\tau^\ast = \argmax_{k = 1, \ldots, n_b} \left| \sum_{j = 1}^{k} \hat{q}_{j,m} - \frac{k}{n_b} \sum_{j = 1}^{n_b} \hat{q}_{j,m} \right|.
\end{equation}
This provides the estimated change point in the block domain, which can be translated into the original time scale by multiplying it by the block size, i.e., $\tau = m \cdot \tau^\ast$.
We refer to~\cite{betken2025ordinal} for a rigorous analysis of the statistical properties of the turning rate estimator, while Algorithm~\ref{alg:cpd-distr-ptxt} summarizes the procedure.

\begin{algorithm}
\caption{Change-Point Detection via Ordinal Patterns}
\label{alg:cpd-distr-ptxt}
\begin{algorithmic}[1]
\Require $x_1, x_2, \dots, x_{n}$ time series, $m$ block size
\Ensure $\tau$ change-point
\vspace{1pt}
\State $n_b \gets \lceil n/m \rceil$ \Comment{Number of blocks}
\For{$j = 1$ to $n_b$}
  \State $\hat{q}_{j,m} \gets 0$ \Comment{Initialize block statistic}
  \For{$i = 1$ to $m-2$}
    \State $s \gets (j - 1) \cdot m + i$
    \State $\pi \gets \Pi(x_s,\, x_{s+1},\, x_{s+2})$
    \If{$\pi \in \mathcal{T}$}
      \State $\hat{q}_{j,m} \gets \hat{q}_{j,m} + 1$
    \EndIf
  \EndFor
  \State $\hat{q}_{j,m} \gets \hat{q}_{j,m} / m$
\EndFor
\For{$k = 1$ to $n_b$}
  \State $S_k \gets \left| \sum_{j=1}^{k} \hat{q}_{j,m} - \frac{k}{n_b} \sum_{j=1}^{n_b} \hat{q}_{j,m} \right|$
\EndFor
\State \Return $\tau \gets m \cdot \argmax_{k} S_k$ \Comment{see Equation~\ref{cusum}}
\end{algorithmic}
\end{algorithm}

\subsection{FHE and CKKS}
\label{sec:background:fhe}

FHE is a cryptographic primitive that allows arbitrary computations to be performed directly on encrypted data.
This enables a client to outsource computations to untrusted servers without revealing any information about the underlying plaintexts.  
CKKS, in particular, is an FHE scheme introduced by Cheon et al.~\cite{cheon2017homomorphic} and designed to work with vectors of real (floating-point) numbers.
CKKS natively supports three operations on ciphertexts: component-wise addition ($X + Y$), component-wise multiplication ($X \cdot Y$), and vector rotations (left $X \ll k$ and right $X \gg k$).
Throughout the paper, we denote ciphertexts with uppercase letters (e.g., $X$ is the encryption of $x$).  

Besides vectors, CKKS can also handle matrices, typically by encoding them row-by-row.
Additions, multiplications, and rotations can be combined to define operations on matrices.
In particular, for a matrix $x$ encrypted as $X$:
\begin{itemize}
    \item $\sumR(X)$ sums all rows of $x$ into the first row;
    \item $\replR(X)$ replicates the first row across all rows;
    \item $\transR(X)$ transposes the first row into the first column.
\end{itemize}
Similar operations are available for the columns, i.e. $\sumC(X)$, $\replC(X)$, $\transC(X)$.
These algorithms can be implemented recursively to improve efficiency, but require padding the input to the next power of two.
Their pseudocode can be found in Appendix~\ref{app:enc-matrix-operations}.

\subsubsection*{Encrypted Ranking and Argmax}

Computing the ordinal patterns (or ranking) $\Pi$ and the $\argmax$ of a vector under homomorphic encryption is non-trivial, since these operations inherently involve comparisons (e.g., $X > Y$), which are extremely expensive under CKKS.
To compute these operations, we build upon the design of Mazzone et al.~\cite{mazzone2025efficient}, which we adapt and optimize to the CPD setting.  
The core idea is as follows.
Given an input vector $\vector$, we first encode it into the first row of a null matrix $X$.
Two re-encodings are then produced: one where values are replicated across rows, $X_R = \replR(X)$, and one where they are replicated across columns, $X_C = \replC(\transR(X))$.
For example, for $\vector = (\vector_1, \vector_2, \vector_3, \vector_4)$ we obtain
$$
\vector_R =\left[
\begin{array}{cccc}
\vector_1 & \vector_2 & \vector_3 & \vector_4 \\
\vector_1 & \vector_2 & \vector_3 & \vector_4 \\
\vector_1 & \vector_2 & \vector_3 & \vector_4 \\
\vector_1 & \vector_2 & \vector_3 & \vector_4 \\
\end{array} \right]
\quad
\vector_C =
\left[
\begin{array}{cccc}
\vector_1 & \vector_1 & \vector_1 & \vector_1 \\
\vector_2 & \vector_2 & \vector_2 & \vector_2 \\
\vector_3 & \vector_3 & \vector_3 & \vector_3 \\
\vector_4 & \vector_4 & \vector_4 & \vector_4 \\
\end{array}\right] \enspace .
$$
A component-wise comparison between $X_R$ and $X_C$ provides the result of the comparison of all pairs $(\vector_i, \vector_j)$.

To perform encrypted comparison, we employ the polynomial approximation of the sign function provided by Cheon et al.~\cite{cheon2020efficient}, which is also used in~\cite{mazzone2025efficient}.
Specifically, we define two polynomials
\begin{align*}
    f(x) &= (35 x - 35 x^3 + 21 x^5 - 5 x^7) / 2^4 \\
    g(x) &= (4589 x - 16577 x^3 + 25614 x^5- 12860 x^7) / 2^{10} \enspace .
\end{align*}
and compose them a given number of times to obtain $\cmp(x,y) = (f^{d_f}(g^{d_g}(x-y)) + 1) / 2$, where $d_f, d_g$ are the number of times each polynomial is composed with itself.
These polynomials can be evaluated using a proper composition of homomorphic additions and multiplications.
For sufficiently high $d_f, d_g$, we have $\cmp(x,y) \approx 1$ if $x > y$, $\cmp(x,y) \approx 0$ if $x < y$, and $\cmp(x,y) = 0.5$ if $x=y$.  
Applying $\cmp$ to $\vector_R$ and $\vector_C$ results in a comparison matrix where each column contains as many entries close to one as the rank of the corresponding element (up to an offset of 0.5 due to self-comparisons).
Hence, summing the rows and correcting by $0.5$ provides the encrypted rank (see Algorithm~\ref{alg:ranking}).

\begin{algorithm}[!htb]
\caption{$\rankAlg$ from~\cite{mazzone2025efficient}}
\label{alg:ranking}
\begin{algorithmic}[1]
\Require $\encVector$ encryption of $\vector = (\vector_1, \dots, \vector_\vectorLength)$
\Ensure $\encRanking$ encryption of the ranking of $\vector$
\State $\encVector_R \gets \replR(\encVector)$
\State $\encVector_C \gets \replC(\transR(\encVector))$
\State $\encComparison \gets \cmp(\encVector_R, \encVector_C)$
\State $\encRanking \gets \sumR(\encComparison) + (0.5, \dots, 0.5)$
\State \Return $\encRanking$
\end{algorithmic}
\end{algorithm}

To compute the $\argmax$ of $N$ encrypted values, we first compute their ranking using Algorithm~\ref{alg:ranking}.
We then apply an indicator function $\ind_\vectorLength$, defined as $\ind_\vectorLength(x)=1$ if $x=\vectorLength$, and $0$ otherwise.
This function can be implemented as a product of two comparisons.
Applying the indicator function to the output of the ranking algorithm returns a one-hot encoding of the argmax of the input vector (see Algorithm~\ref{alg:argmax}).
In Section~\ref{sec:secureCPD:cusum}, we describe a more efficient argmax design that avoids computing the full ranking.  

\begin{algorithm}[!htb]
\caption{$\argmaxAlg$ from~\cite{mazzone2025efficient}}
\label{alg:argmax}
\begin{algorithmic}[1]
\Require $\encVector$ encryption of $\vector = (\vector_1, \dots, \vector_\vectorLength)$
\Ensure $\encArgmax$ encryption of $\argmax(\vector)$ as one-hot encoding
\State $\encRanking \gets \rankAlg(\encVector)$ \Comment{see Algorithm~\ref{alg:ranking}}
\State $\encArgmax \gets \ind_\vectorLength(\encRanking)$
\State \Return $\encArgmax$
\end{algorithmic}
\end{algorithm}

\noindent
As a side note, the above algorithms may not behave correctly when $\vector_i = \vector_j$ for $i \ne j$.
While Mazzone et al.~\cite{mazzone2025efficient} propose an additional tie-handling mechanism, we can avoid its cost here, since ties usually occur with low probability in time-series data.

\section{Secure CPD}

We now present our approach for estimating the change point $\tau$ of a time series under encryption.
The approach can be divided into two steps:
\begin{enumerate}
    \item a block-wise summarization of the time series data, which turns the CPD problem into a mean-shift detection problem, and
    \item an application of the CUSUM statistic to solve the mean-CPD problem.
\end{enumerate}
The exact summarization function depends on the type of change we want to detect.
Mean and variance are used to detect shifts in the corresponding properties, while the turning rate statistic is used for changes in frequency (see Section~\ref{sec:background:turningrate}).
We first describe the matrix encoding we use for the input data, then we provide a secure implementation of these three summarization functions and the CUSUM statistic, and finally we put them together to build CPD algorithms for detecting changes in mean, variance, and frequency.

\subsection{Block-Based Matrix Encoding}
\label{sec:secure-cpd:assumptions-and-matrix-encoding}

Given a time series $x = (x_1, \dots, x_n)$ and a block size $m$ for our summarization function, we define $n_b = \ceil{n / m}$ to be the number of blocks of the series.
We make the following assumptions on the block size to simplify the description of the algorithm:
\begin{enumerate}
    \item $m$ divides $n$, ensuring all blocks have the same size, and
    \item $m$ is a power of two, allowing a straightforward application of the recursive $\sumR$, $\sumC$, $\replR$, $\replC$, $\transR$, $\transC$ algorithms.
\end{enumerate}
Both assumptions can be lifted with minor adjustments to the algorithms, as discussed in Section~\ref{sec:secure-cpd:general-block-sizes}.

A naive implementation of the plaintext algorithms using the tools from Section~\ref{sec:background:fhe} would result in many nested loops of homomorphic operations and a large number of ciphertexts to manage.
As we will see in the following sections, we can significantly reduce the computational cost by fully exploiting the SIMD capabilities of CKKS.
For this purpose, we represent the time series as a matrix:
$$
\begin{bmatrix}
x_1 & x_2 & \dots & x_m \\
x_{m+1} & x_{m+2} & \dots & x_{2m} \\
\vdots & \vdots &  & \vdots \\
x_{n - m + 1} & x_{n - m + 2} & \dots & x_n \\
\end{bmatrix}
$$
where each row corresponds to a block of $m$ elements.
We also design the summarization functions so that each block's summary is stored in the first element of the corresponding row.
The CUSUM computation is then designed accordingly to take as input the first column of this matrix representation (see Figure~\ref{fig:cpd-pipeline}).
The reason behind this choice will become clear during the description of the different components.

\begin{figure}
    \centering
    \begin{tikzpicture}[node distance=40pt]
    
    \tikzstyle{every node}=[font=\footnotesize]
    \renewcommand{\tabcolsep}{\tabcolsepValue}
    \renewcommand{\arraystretch}{\arraystretchValue}

    \node (ts) {
    \begin{tabular}{|P{10pt}|P{10pt}|P{10pt}|}
        \hline
        $x_1$ & $x_2$ & $x_3$ \\ \hline
        $x_4$ & $x_5$ & $x_6$ \\ \hline
        $x_7$ & $x_8$ & $x_9$ \\ \hline
    \end{tabular}
    };

    \node[right=of ts] (summary) {
    \begin{tabular}{|P{10pt}|P{10pt}|P{10pt}|}
        \hline
        $s_1$ &  &  \\ \hline
        $s_2$ &  &  \\ \hline
        $s_3$ &  &  \\ \hline
    \end{tabular}
    };

    \node[right=of summary] (cusum) {
    \begin{tabular}{|P{10pt}|P{10pt}|P{10pt}|}
        \hline
        0 & 1 & 0 \\ \hline
         &  &  \\ \hline
         &  &  \\ \hline
    \end{tabular}
    };

    \node[above=1pt of ts] {\strut \textbf{Time Series}};
    \node[above=1pt of summary] {\strut \textbf{Summarized Blocks}};
    \node[above=-1pt of cusum, align=center] {\strut \textbf{Change Point} \\ \textbf{(One-Hot)}};

    \draw[->, thick] (ts) -- (summary) node[midway, above] {Summarize};
    \draw[->, thick] (summary) -- (cusum) node[midway, above] {CUSUM};

    \end{tikzpicture}
    \caption{Schematic representation of our secure CPD pipeline. Here $n = 9$, $m = 3$, and the change point is detected at the second block.}
    \label{fig:cpd-pipeline}
    \Description{}
\end{figure}
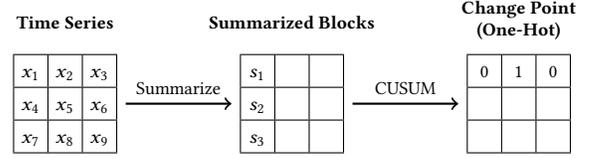

\subsection{Summarization Functions}
\label{sec:summ_functions}

Implementing efficient block-wise mean and variance is relatively straightforward, while the ordinal pattern representation is more challenging.

\subsubsection*{Mean}

Computing the mean of each block can be done efficiently with a $\sumC$ operation, which recursively sums all elements of each block into the first column of the matrix, and a scalar multiplication by $1/m$ (see Algorithm~\ref{alg:enc-mean}).

\begin{algorithm}
\caption{$\meanAlg$}
\label{alg:enc-mean}
\begin{algorithmic}[1]
\Require $X$ encryption of $x = (x_1, \dots, x_n)$, $m$ block size
\Ensure $S$ encryption of $s^\intercal = (s_1, \dots, s_{n_b})^\intercal$ block-wise mean of $x$
\State $S \gets \sumC(X)$
\State $S \gets (1/m) \cdot S$
\State \Return $S$
\end{algorithmic}
\end{algorithm}

\subsubsection*{Variance}

Computing the variance is more complex, but it can still be done efficiently.
We recall that for each block $x_{km + 1}, \dots, x_{(k+1)m}$ we need to compute
$$\frac{1}{m-1} \sum_{i = 1}^{m}{(x_{km + i} - \bar{x}_k)^2}$$
where $\bar{x}_k$ is the mean of the given block.
First, we compute the mean as in Algorithm~\ref{alg:enc-mean} and we replicate it over the columns with a $\replC$.
Then, we subtract it from the original time series matrix, square the result, sum it block-wise with a $\sumC$, and multiply it by $\frac{1}{m-1}$ (see Algorithm~\ref{alg:enc-variance}).

\begin{algorithm}
\caption{$\varAlg$}
\label{alg:enc-variance}
\begin{algorithmic}[1]
\Require $X$ encryption of $x = (x_1, \dots, x_n)$, $m$ block size
\Ensure $S$ encryption of $s^\intercal = (s_1, \dots, s_{n_b})^\intercal$ block-wise variance of $x$
\State $M \gets \meanAlg(X, m)$ \Comment{Algorithm~\ref{alg:enc-mean}}
\State $M \gets \replC(M)$
\State $S \gets 1/(m-1) \cdot \sumC((X - M)^2)$
\State \Return $S$
\end{algorithmic}
\end{algorithm}

\subsubsection*{Turning Rates}

For each block we create a sliding window of size three and count how many triplets have rank in $\mathcal{T} $.
A straightforward implementation would consist of mirroring Algorithm~\ref{alg:cpd-distr-ptxt}.
For each block $j = 1, \dots, n_b$:
\begin{enumerate}
    \item Triplet extraction:
        extract $m - 2$ triplets of consecutive values $(x_i,\, x_{i+1},\, x_{i+2})$ by applying plaintext bitmasks and ciphertext rotations that zero out irrelevant entries and isolate each~triplet.
    \item Ordinal pattern computation: compute the ordinal pattern of each triplet using Algorithm~\ref{alg:ranking}, resulting in $R_i = (r_i, r_{i+1}, r_{i+2})$.
    \item Pattern membership check:
        determine whether the pattern belongs to $\mathcal{T}$.
        This can be done for instance by evaluating
        $$\qquad\quad M_i = \frac{1}{12} (\norm{R_i - (1,2,3)}_2^2 + \norm{R_i - (3,2,1)}_2^2)$$
        which equals one if $R_i \in \mathcal{T}$ and zero otherwise, since the only patterns not in $\mathcal{T}$ are $(1,2,3)$ and $(3,2,1)$ (one can easily verify this by checking all six possible permutations).
    \item Turning rate estimation:
        the number of matching patterns is computed by summing all $M_i$; then it is normalized by the number of triplets, yielding an encrypted estimate of the turning rate:
        $$\qquad\quad \enc(\hat{q}_{j,m}) = \enc(\# \text{matches} / (m - 2))\enspace .$$
\end{enumerate}

This implementation employs $n_b (m - 2)$ ciphertexts and requires as many $\rankAlg$ operations, making it particularly inefficient.
Each call to $\rankAlg$ uses only $3^2 = 9$ slots, leaving most ciphertext slots unused (a ciphertext typically has $2^{13}$ to $2^{17}$ slots).
Parallelizing multiple ranking operations within one ciphertext would be a valid optimization, but a more efficient approach is possible.
To verify whether a triplet $(x_i,\, x_{i+1},\, x_{i+2})$ belongs to $\mathcal{T}$, we do not actually need all 9 comparisons $x_i > x_j$ performed in $\rankAlg$.
It is enough to check whether the triplet is monotonic, increasing or decreasing, and declare membership to $\mathcal{T}$ otherwise.
This happens if and only~if
$$x_i < x_{i+1} < x_{i+2} \quad\text{or}\quad x_i > x_{i+1} > x_{i+2} \enspace .$$
Hence, we potentially need only four comparisons instead of nine.
Furthermore, we observe that $(x_i < x_{i+1}) = 1 - (x_i > x_{i+1})$, and similarly $(x_{i+1} < x_{i+2}) = 1 - (x_{i+1} > x_{i+2})$.
Hence, membership of a triplet in $\mathcal{T}$ can be verified using only two comparisons:
$$c_i := (x_i > x_{i+1}), \qquad c_{i+1} := (x_{i+1} > x_{i+2}) \enspace .$$
In particular, the triplet $(x_i,\, x_{i+1},\, x_{i+2})$ is not a member of $\mathcal{T}$ if and only if either $c_i, c_{i+1}$ are both zero or both one.
The membership condition corresponds to the XNOR $\overline{c_i \oplus c_{i+1}}$, which we can write in arithmetic form as
\begin{equation}
\label{eq:memb-arithmetic-condition}
\overline{c_i \oplus c_{i+1}} = c_i + c_{i+1} - 2 c_i c_{i+1} = 1 - (c_i + c_{i+1} - 1)^2 \enspace .
\end{equation}
Finally, we note that the next triplet $(x_{i+1},\, x_{i+2},\, x_{i+3})$ requires us to compute $c_{i+1} = (x_{i+1} > x_{i+2})$ again, which we can recycle from the current triplet:
\newcommand{\negspace}{\hspace{-12pt}}
\begin{align*}
(x_i,\, x_{i+1},\, x_{i+2}) &\quad\text{requires}\quad (x_i > x_{i+1}) &&\text{and}\quad(x_{i+1} > x_{i+2}) \\[3pt]
(x_{i+1},\, x_{i+2},\, x_{i+3}) &\quad\text{requires}\quad (x_{i+1} > x_{i+2})&&\text{and}\quad(x_{i+2} > x_{i+3}) \\[3pt]
(x_{i+2},\, x_{i+3},\, x_{i+4}) &\quad\text{requires}\quad (x_{i+2} > x_{i+3}) &&\text{and}\quad(x_{i+3} > x_{i+4}) \\[-4pt]
\vdots \hspace{28pt} &\hspace{70pt} \vdots \hspace{30pt} &&\hspace{45pt}\vdots
\end{align*}
Thus, to evaluate all triplets within a block $x_{km + 1}, \dots, x_{(k+1)m}$, it is enough to compute all the adjacent comparisons:
$$(x_{km+1} > x_{km+2}), \dots, (x_{(k+1)m-1} > x_{(k+1)m}) \enspace .$$
This can be done efficiently in one batch by $C = \cmp(X,\allowbreak X \ll 1)$, which operates on every block concurrently.
To determine membership we proceed as described in Equation~\ref{eq:memb-arithmetic-condition} by computing
$$F = 1 - (C + (C \ll 1) - 1)^2$$
which produces a matrix
$$
f_{k,i} = \begin{cases}
    1 & \text{if } r_{k,i} \in \mathcal{T} \\
    0 & \text{otherwise}
\end{cases}
$$
where $r_{k,i} = \Pi(x_{km + i},\, x_{km + i + 1},\, x_{km + i + 2})$.
Finally, by summing each column with $\sumC$ we can count the number of triplets in $\mathcal{T}$ per block, resulting in a column vector $(q_1, \dots, q_{n_b})^\intercal$ such that
$$q_k = \sum_{i = 1}^{m}{f_{k,i}} \enspace .$$
The turning rates are then computed by a plaintext multiplication by $1/(m-2)$.
Algorithm \ref{alg:enc-turning-rates} summarizes these steps in pseudocode.

\begin{algorithm}
\caption{$\turnRatesAlg$}
\label{alg:enc-turning-rates}
\begin{algorithmic}[1]
\Require $X$ encryption of $x = (x_1, \dots, x_n)$, $m$ block size
\Ensure $Q$ encryption of $q^\intercal = (q_1, \dots, q_{n_b})^\intercal$ block-wise turning rates of $x$
\State $C \gets \cmp(X, X \ll 1)$
\State $F \gets 1 - (C + (C \ll 1) - 1)^2$ \Comment{pattern membership}
\State $Q \gets \frac{1}{m-2} \cdot \sumC(F)$ \Comment{block-wise turning rate}\vspace{2pt}
\State \Return $Q$
\end{algorithmic}
\end{algorithm}

\subsection{CUSUM Statistic}
\label{sec:secureCPD:cusum}

The block-wise summaries are now stored in the first column of a matrix as $s_1, \dots, s_{n_b}$ and we need to compute
\begin{equation}
\label{eq:cusum-simple}
\argmax_{k = 1, \ldots, n_b} \left| \sum_{j = 1}^{k} s_j - \frac{k}{n_b} \sum_{j = 1}^{n_b} s_j \right| \enspace .
\end{equation}
There are four main operations we need to consider here:
\begin{itemize}
    \item the total sum $s_T := \sum_{j = 1}^{n_b} s_j$,
    \item the partial sum $s_P^{(k)} := \sum_{j = 1}^{k} s_j$ for each $k$,
    \item the absolute value, and
    \item the argmax computation.
\end{itemize}
We address each operation as follows.

\subsubsection*{Total Sum}

Computing the total sum $s_T$ is straightforward.
Since the $s_i$ are stored as elements of the first column, it is enough to execute a $\sumR$ and the resulting $s_T$ will end up in the top-left entry of the matrix.
In preparation to the argmax computation over $k$ elements, we replicate $s_T$ across the first row using a $\replC$ operation, and then multiply element-wise by the plaintext vector
$$\left(\frac{1}{n_b}, \frac{2}{n_b}, \dots, \frac{n_b}{n_b}\right)$$
to scale it according to Equation~\ref{eq:cusum-simple}.

\subsubsection*{Partial Sums}

Computing the partial sum $s_P^{(k)}$ for each $k$ is more challenging.
A naive solution would compute each $s_P^{(k)}$ individually by masking all elements but the first $k$ and summing them with a $\sumC$, but this approach would split the operations over $n_b$ ciphertexts.
Instead, we can process all the partial sums in one go.
\begin{enumerate}
    \item First, we replicate the input column vector over all columns using $\replC$.
    \item We then multiply the resulting matrix by an upper-triangular bitmask $T = (t_{i,j})$, where $t_{i,j} = 1$ if $j \ge i$ and $0$ otherwise.
    This mask ensures that the first column contains only $s_1$, the second column contains $s_1$ and $s_2$, and in general the $k$-th column contains $s_1,\allowbreak \dots,\allowbreak s_k$.
    \item Finally, a $\sumR$ aggregates the elements in each column, resulting in the vector of partial sums $(s_P^{(1)},\allowbreak s_P^{(2)},\allowbreak \dots,\allowbreak s_P^{(n_b)})$ in the first row of the matrix.
\end{enumerate}

\subsubsection*{Absolute Value}
The absolute value is not a polynomial function, hence it is not straightforward to implement under CKKS.
The usual way is to compute it as $|s| = \sqrt{s^2}$, using an iterative method to approximate the square root~\cite{cheon2019numerical}.
Fortunately, we do not need to compute the absolute values explicitly.
Since $\argmax_k |s_k| = \argmax_k s_k^2$, we can square the values instead and compare them directly, avoiding the expensive square root approximation.

\subsubsection*{Argmax}
The argmax can be computed using the approach described in Section~\ref{sec:background:fhe} (see Algorithm~\ref{alg:argmax}).
This involves computing the comparison matrix $(s_i > s_j)$ for all $i,j$, summing over the rows to obtain the rankings, and applying an indicator function to select the desired rank.
However, in cases where we are specifically interested in the argmax, we can significantly improve this approach.
We observe that the column of the comparison matrix corresponding to the maximum entry contains only ones (and a $0.5$ on the diagonal), while all other columns contain at least one zero.
This means that instead of computing a full ranking and applying an expensive indicator function, we can directly isolate the argmax position by multiplying all values across the rows, that is computing a product down each column.
Only the column corresponding to the maximum will result in a non-zero product, producing a one-hot encoding of the argmax.

This method allows us to replace the indicator function with a much cheaper recursive multiplication over rows, requiring only a logarithmic number of homomorphic rotations and multiplications (think of it as the multiplicative variant of $\sumR$).
To give a practical idea of the speed-up, consider the case of $n_b = 256$.
In our experiments, the $\cmp$ operation costs around $\nu = \rho = 25$ multiplications ($\nu$) and circuit depth ($\rho$).
Hence, the baseline method (comparison followed by $\ind$) requires $2\nu + 1 = 51$ extra multiplications and $\rho + 1 = 26$ extra circuit depth, for a total of $76$ multiplications and $51$ circuit depth.
On the other hand, our optimized method only adds $\log_2(n_b) = 8$ multiplications and depth on top of the comparison, for a total of $33$ multiplications and circuit depth, resulting in a reduction of \mytilde57\% in the number of multiplications and \mytilde35\% in circuit depth.

We put these four operations together in Algorithm~\ref{alg:enc-cusum} to compute the CUSUM statistic provided in Equation~\ref{eq:cusum-simple} under encryption.

\begin{algorithm}
\caption{$\cusumAlg$}
\label{alg:enc-cusum}
\begin{algorithmic}[1]
\Require $S$ encryption of $s^\intercal = (s_1, \dots, s_{n_b})^\intercal$
\Ensure $C$ encryption of a one-hot encoding of the change-point location $\tau$ (block-wise)
\Statex \textbf{\textit{Partial and total sums}}
\State $S_T \gets \replC(\sumR(S))$
\State $S_P \gets \sumR(\replC(S) \cdot \triu(1))$
\Statex \Comment{$\triu(1)$ is the upper-triangular binary mask}
\State $U \gets (S_P - (1/n_b, 2/n_b, \dots, n_b/n_b) \cdot S_T)^2$
\Statex \textbf{\textit{Our revised argmax}}
\State $U_R \gets \replR(U)$
\State $U_C \gets \replC(\transR(U))$
\State $C = \cmp(U_R, U_C) + \diag(0.5)$
\Statex \Comment{$\diag(0.5)$ is the diagonal matrix with $0.5$ entries}
\For{$i = 0, \dots, \log{n_b} - 1$}
    \State $C \gets C \cdot (C \ll m \cdot 2^i)$
\EndFor
\State \Return $C$
\end{algorithmic}
\end{algorithm}

\subsection{CPD Algorithms}

We now combine the summarization algorithms with the CUSUM statistic to build our full CPD pipeline.
Depending on the type of change to detect (mean, variance, or frequency) we use the appropriate summarization function:
\begin{itemize}
    \item mean: uses block-wise mean (see Algorithm~\ref{alg:enc-mean}),
    \item variance: uses block-wise variance (see Algorithm~\ref{alg:enc-variance}),
    \item frequency: uses block-wise turning rate distribution (see Algorithm~\ref{alg:enc-turning-rates}).
\end{itemize}
The output summary is then processed with the encrypted CUSUM algorithm (Algorithm~\ref{alg:enc-cusum}) to obtain the final result.
See Algorithm~\ref{alg:enc-cpd} for the complete pseudocode of the encrypted CPD pipeline.
Note that the resulting change-point index $c_b$ is expressed in block units.
To recover the index in the original time series, it must be scaled by the block size: $c = c_b \cdot m$.

\begin{algorithm}[H]
\caption{$\cpdAlg$}
\label{alg:enc-cpd}
\begin{algorithmic}[1]
\Require $X$ encryption of $x = (x_1, \dots, x_n)$, $m$ block size, $t$~change type (mean, variance, frequency)
\Ensure $C$ encryption of a one-hot encoding of the change-point location $\tau$ (block-wise)
\If{$t = \texttt{mean}$}
    \State $Q \gets \meanAlg(X)$ \Comment{see Algorithm~\ref{alg:enc-mean}}
\ElsIf{$t = \texttt{variance}$}
    \State $Q \gets \varAlg(X)$ \Comment{see Algorithm~\ref{alg:enc-variance}}
\ElsIf{$t = \texttt{frequency}$}
    \State $Q \gets \turnRatesAlg(X)$ \Comment{see Algorithm~\ref{alg:enc-turning-rates}}
\EndIf
\State $C \gets \cusumAlg(Q)$ \Comment{see Algorithm~\ref{alg:enc-cusum}}
\State \Return $C$
\end{algorithmic}
\end{algorithm}

\noindent
Two technical details still need to be addressed though: the normalization bounds for approximate comparisons, and the support for general values of $m$ (lifting the assumptions we made in Section~\ref{sec:secure-cpd:assumptions-and-matrix-encoding}).

\subsubsection*{Normalizing Input of \textsf{cmp}}

The homomorphic evaluation of $\cmp$ relies on a polynomial approximation of the sign function, approximated in the interval $[-1,1]$.
Hence, to make this work correctly, we need to ensure that all inputs passed to the $\cmp$ function are normalized in $[0,1]$.
Since the cost of evaluating the $\cmp$ depends on the precision of the approximation, obtaining the tightest possible bounds on its inputs is important for efficiency.

There are two $\cmp$ operations in our algorithm: one in the turning rate computation and one in the CUSUM algorithm.
The first $\cmp$ is fine as it is directly applied to the input data, which we can simply normalize to $[0,1]$ under the (mild) assumption that upper and lower bounds of the series are known (e.g., heart rate between 0 and 600 bpm, EEG between -200 and +200{\textmu}V).
The second one (CUSUM) is more challenging.
In this case, the input to the $\cmp$ are the values
$$ \left| \sum_{j = 1}^{k} s_j - \frac{k}{n_b} \sum_{j = 1}^{n_b} s_j \right| \qquad \text{for } k = 1, \dots, n_b \enspace .$$
We are providing some bounds on $\Delta_k :=  \sum_{j = 1}^{k} s_j - \frac{k}{n_b} \sum_{j = 1}^{n_b} s_j$.
First, as the time series values $x_j$ are assumed to be normalized in $[0,1]$, then the $s_j$ must also be in $[0,1]$:
\begin{itemize}
    \item for the mean, this is trivial;
    \item for the variance, you can actually see that $s_j \in [0, 1/4]$;
    \item for the turning rates, $s_j$ is the number of triplets in $\mathcal{T}$ divided by the number of total triplets, hence again $s_i \in [0,1]$.
\end{itemize}
Since $s_T \ge s_P^{(k)}$, we have that
\begin{equation*}
\Delta_k = s_P^{(k)} - \frac{k}{n_b} s_T \le s_P^{(k)} - \frac{k}{n_b} s_P^{(k)} = s_P^{(k)} \left(1 - \frac{k}{n_b}\right) \enspace .
\end{equation*}
As $s_j \in [0,1]$, we have that \( s_P^{(k)} \le k \), and thus
\begin{equation*}
\Delta_k \le k\left(1 - \frac{k}{n_b}\right) = k - \frac{k^2}{n_b} \enspace .
\end{equation*}
The right hand side attains maximum at \( k = n_b / 2 \), hence
\begin{equation*}
\Delta_k \le \frac{n_b}{2} - \frac{(n_b/2)^2}{n_b} = \frac{n_b}{2} - \frac{n_b}{4} = \frac{n_b}{4} \enspace .
\end{equation*}
Now, for the lower bound, we can rewrite $-\Delta_k$ as:
\begin{equation*}
    -\Delta_k = \frac{k}{n_b} \sum_{j = 1}^{n_b} s_j - \sum_{j = 1}^{k} s_j = \left(\frac{k}{n_b} - 1\right) \sum_{j = 1}^{k} s_j + \frac{k}{n_b} \sum_{j = k + 1}^{n_b} s_j \enspace .
\end{equation*}
Since $s_j \in [0,1]$ and $k \le n_b$, we have that $\left(\frac{k}{n_b} - 1\right) \sum_{j = 1}^{k} s_j$ is a negative term and $\frac{k}{n_b} \sum_{j = k + 1}^{n_b} s_j$ is a positive term.
Hence, $-\Delta_k$ is maximized when $s_j = 0$ for $j \le k$ and $s_j = 1$ for $j > k$.
So,
\begin{equation*}
    -\Delta_k \le \left(\frac{k}{n_b} - 1\right) \sum_{j = 1}^{k} 0 + \frac{k}{n_b} \sum_{j = k + 1}^{n_b} 1 = \frac{k}{n_b} (n_b - k) = k - \frac{k^2}{n_b} \enspace .
\end{equation*}
Following the same argument as above, we have that $-\Delta_k \le n_b / 4$.
Putting everything together leads to $|\Delta_k| \le n_b / 4$, which provides a bound we can use to normalize $\Delta_k$ before the argmax is computed.

\subsubsection*{Handling Blocks of General Sizes}
\label{sec:secure-cpd:general-block-sizes}

In Section~\ref{sec:secure-cpd:assumptions-and-matrix-encoding} we assumed that the block size $m$ divides the total data length $n$ and is a power of two.
If $m$ does not divide $n$, then the last block has length $m' < m$.
In this case, the series is padded with $m - m'$ zeros, and the summarization functions require some minor adjustments.
For mean and variance, the normalization by $m$ becomes a multiplication by $(\frac{1}{m}, \dots, \frac{1}{m}, \frac{1}{m'})^\intercal$.
For turning rates, a masking is applied to $F$ to mask out the first $m'-2$ elements in the last row; and similarly, the normalization step becomes a multiplication by $(\frac{1}{m-2}, \dots, \frac{1}{m-2}, \frac{1}{m'-2})^\intercal$.

If $m$ is not a power of two, the straightforward solution is to pad each row of the matrix with zeros until a power of two is reached and then adopting similar corrections as above on all rows.
Although, this way we might end up with up to 50\% of unused slots.
To avoid this, all recursive operations on matrices $\sumR$, $\sumC$, $\replR$, $\replC$, $\transR$, $\transC$ can be easily adapted to work recursively on input sizes that are not power of two, using appropriate branching, with their cost going from $\ceil{\log(m)}$ to $2 \floor{\log(m)}$.
See for instance Section~3.1 and Appendix~B of \cite{mazzone2025privacy}.

\section{Experimental Evaluation}
\label{sec:experiments}

We evaluate our encrypted CPD approach on both synthetic and real-world time series, reporting runtime, memory, and accuracy.

\subsection{Experimental Setup}

We use the CKKS implementation from the OpenFHE library~\cite{openFHE},\footnote{\url{https://github.com/openfheorg/openfhe-development}} with a scaling factor ranging between 44 and 50 bits.
The ring dimension is set to $2^{17}$, and all parameters are chosen in accordance with the Homomorphic Encryption Standard to guarantee 128-bit security~\cite{albrecht2015concrete, HomomorphicEncryptionSecurityStandard}.
Our implementation is publicly available at \repository.
To compute ranking and $\argmax$, we employ the $f,g$ approximation of the sign function described in Section~\ref{sec:background:fhe}, with composition degrees $d_f = 2$ and $d_g = 4$.
Since the multiplicative depth of our circuit is bounded by 65, bootstrapping is not required, and CKKS is used as a leveled scheme.
Following Betken et al.~\cite{betken2025ordinal}, we set the block size to $m = \floor{\sqrt{n}}$.
All experiments are executed on a Linux machine equipped with an AMD EPYC 7763 64-Core Processor running at 2.45GHz, and 512~GB of RAM.

\subsection{Synthetic Data}

To evaluate our method under controlled conditions, we use synthetic data in which a change point is explicitly introduced.
Given a length $n$ and a distribution \(\text{Dist}\), we generate a series $(X_1, \dots, X_n)$ with a change point at $\tau$ as follows.

\subsubsection*{Mean Shift}
Given $\mu_1 \neq \mu_2$ and $\sigma=\sigma_1=\sigma_2$:
\begin{align*}
    &X_t\sim  \text{Dist}(\mu_1, \sigma^2)\quad\text{for }\,\, t=1, \ldots,\tau \\
    &X_t\sim  \text{Dist}(\mu_2, \sigma^2)\quad\text{for }\,\, t=\tau+1, \ldots,n 
\end{align*}

\subsubsection*{Variance Shift}
Given $\mu=\mu_1 =\mu_2$ and $\sigma_1\neq \sigma_2$:
\begin{align*}
    &X_t\sim  \text{Dist}(\mu, \sigma_1^2)\quad\text{for }\,\, t=1, \ldots,\tau \\
    &X_t\sim  \text{Dist}(\mu, \sigma_2^2)\quad\text{for }\,\,t=\tau+1, \ldots,n 
\end{align*}
        
\subsubsection*{Frequency Shift}
Given $\mu$, $\sigma$, and $\phi_1 \ne \phi_2$, we simulate data from an autoregressive process of order one (AR(1)):
\begin{align*}
    &X_t= \phi_1 X_{t-1} + \varepsilon_t, \quad \quad\text{for }\,\, t=1, \ldots,\tau \\
    &X_t= \phi_2 X_{t-1} + \varepsilon_t, \quad \quad\text{for }\,\, t=\tau+1, \ldots,n
\end{align*}
where $\varepsilon_t \sim \text{Dist}(\mu, \sigma^2)$ forms a sequence of mutually uncorrelated variables $(\varepsilon_t)_t$, the parameters \(\phi_1, \phi_2 \in (-1,1)\) control the strength of the temporal dependence, and \(X_0 \) is fixed.

\medskip

In our experiments, the change point is placed at $\tau = \floor{n/2}$, \(\text{Dist}\) is selected from Gaussian, Uniform, Laplace, and Student-\(t\) distributions, and the AR(1) parameters are set to $\phi_1 = 0.3$ and $\phi_2 = 0.7$.

\subsection{Real-World Data}
\label{sec:eval:real-data}

We showcase the practical applicability of our approach on three real-world datasets.
They all contain frequency shifts, which represent the most common type of change-points observed in practice.

\subsubsection*{Sleep Phase Change Detection in EEG}
Electroencephalogram (EEG) signals measure brain activity and are widely used in sleep studies to identify transitions between sleep stages.
These transitions are typically reflected by changes in the signal frequency.  
For our evaluation, we use data from the CAP Sleep Database by Terzano et al.~\cite{terzano2001cap}, available through PhysioBank.\footnote{\url{https://physionet.org/content/capslpdb/1.0.0/}}
Following Bandt~\cite{bandt2020order}, we consider the overnight recording of subject~5 from the C4–P4 channel.
The extracted series contains 60{,}000 samples (about four minutes at 256~Hz) and is annotated with expert-provided sleep stage labels, which we use as ground truth.

\subsubsection*{Heartbeat Interval Time Series}
Inter-beat interval (IBI) signals measure the time between successive heartbeats and are commonly used as a non-invasive proxy for autonomic nervous system activity.  
For our assessment, we use data from the Meditation Task Force dataset~\cite{peng1999meditation}, available via PhysioNet.\footnote{\url{https://physionet.org/physiobank/database/meditation/data/}}
During meditation, changes in the IBI series indicate shifts in heart rate frequency, linked to transitions in relaxation and stress levels.
Specifically, we analyze a one-hour meditation recording from a single subject, which yields an IBI time series of length 8{,}083.

\subsubsection*{Network Traffic Anomaly Detection}
Network flow time series are widely used to detect unusual activity patterns in communication networks. 
For our experiments, we use the CESNET-TimeSeries24 dataset~\cite{koumar2025cesnet}, a large-scale collection of time series extracted from IP flow records in the CESNET3 academic ISP network.\footnote{\url{https://zenodo.org/records/13382427}}  
The dataset comprises 40 weeks of monitoring over more than 275{,}000 IP addresses and includes volumetric, ratio-based, and temporal statistics (e.g., number of bytes, packets, unique destination IPs, average flow duration).
Changes in frequency in these series correspond to anomalies in traffic behavior.
We employ the per-IP time series of the inbound-to-outbound packet ratio ($\textsf{dir\_ratio\_packets}$), aggregated over 10-minute intervals, yielding sequences of 40{,}298 time points.

\subsection{Empirical Results}

\begin{figure}
  \centering
\begin{tikzpicture}
  \small
  \pgfplotscreateplotcyclelist{mylist}{%
    {blue,   mark=triangle*, mark options={solid}, thick},
    {red,    mark=square*,   mark options={solid}, thick},
    {green!60!black, mark=o, mark options={solid}, thick},
    {orange, mark=+,         mark options={solid}, thick},
  }
  \begin{axis}[
    tick label style={/pgf/number format/fixed},
    xlabel={Series length $n$},
    ylabel={Runtime (s)},
    grid=both,
    major grid style={line width=.2pt, draw=gray!30},
    minor grid style={line width=.1pt, draw=gray!10},
    every axis plot/.append style={line width=1.1pt},
    cycle list name=mylist,
    enlarge x limits=0.02,
    enlarge y limits=0.05,
    xmin=0,
    height=6cm, width=\linewidth,
    legend style={at={(0.5,1.05)}, anchor=south, draw=none, yshift=-2pt},
    legend columns=3, legend cell align={left},
    scaled x ticks=false,
    scaled y ticks=false
  ]
    \input{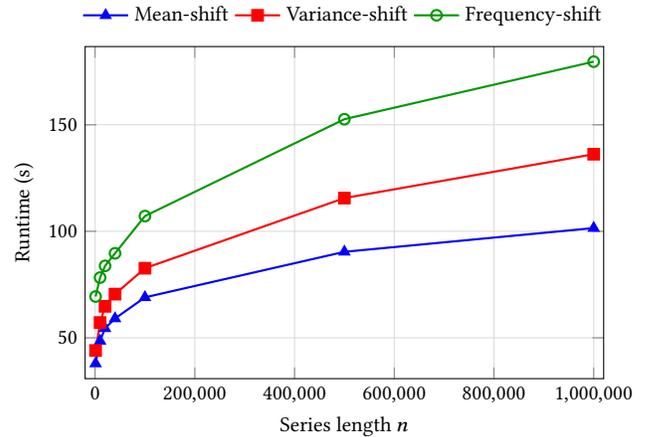}
  \end{axis}
\end{tikzpicture}
\caption{Runtime of our approach for increasing time series length.}
\label{fig:runtime}
\end{figure}

\begin{table*}
    \caption{Total and per-component runtime, communication size, and memory consumption of our approach on synthetic data.}
  \label{tab:runtime-micro}
  \begin{tabular}{lr|rr|rrrrr}
    \toprule
     & & \multicolumn{2}{c|}{\textbf{Client}} & \multicolumn{5}{c}{\textbf{Server}} \\
    Type & Length & Encryption (s) & Comm. (GB) & Turn. Rates (s) & Partial Sums (s) & Argmax (s) & Total (s) & Memory (GB) \\
    \midrule
mean & 1,000 & 0.51 & 0.1 & - & 8.22 & 20.29 & 37.91 & 8.6 \\
mean & 10,000 & 0.55 & 0.1 & - & 11.70 & 24.34 & 48.54 & 11.3 \\
mean & 100,000 & 3.01 & 0.3 & - & 15.47 & 35.06 & 68.99 & 16.1 \\
mean & 1,000,000 & 10.69 & 1.3 & - & 16.52 & 56.98 & 101.52 & 31.6 \\
variance & 1,000 & 0.52 & 0.1 & - & 7.16 & 19.87 & 44.05 & 8.7 \\
variance & 10,000 & 0.56 & 0.1 & - & 10.83 & 22.81 & 57.17 & 11.3 \\
variance & 100,000 & 2.98 & 0.3 & - & 14.78 & 31.37 & 82.67 & 16.4 \\
variance & 1,000,000 & 11.17 & 1.4 & - & 16.66 & 53.35 & 136.20 & 32.3 \\
frequency & 1,000 & 0.85 & 0.1 & 38.16 & 6.09 & 20.34 & 69.38 & 13.2 \\
frequency & 10,000 & 0.86 & 0.1 & 40.11 & 8.63 & 23.54 & 78.27 & 17.4 \\
frequency & 100,000 & 4.92 & 0.5 & 54.67 & 11.69 & 32.95 & 107.16 & 25.6 \\
frequency & 1,000,000 & 17.46 & 2.2 & 108.80 & 12.70 & 51.20 & 179.66 & 48.9 \\
  \bottomrule
\end{tabular}
\end{table*}

\paragraph{Runtime}
In Figure~\ref{fig:runtime}, we report the runtime of our solution on synthetic data with changes in mean, variance, and frequency, for time series of increasing length (from 1,000 to 1,000,000 points).
The reported values are averages over 4 runs, using a Gaussian distribution $\mathcal{N}(0,1)$ as base.
We generate the three types of changes by shifting the mean to $1$, increasing the variance to $2$, and setting $\phi_1 = 0.3$, $\phi_2 = 0.7$ for the frequency change.
As expected, frequency detection takes longer than mean and variance detection, since computing the turning rate statistic requires an additional comparison, which is the bottleneck of our method.
Specifically, detecting changes in mean requires between 37.91s and 101.52s, variance detection between 44.05s and 136.20s, and frequency detection between 69.38s and 179.66s.
Overall, the runtime seems to increase logarithmically with the data size, as expected by the logarithmic scaling of the matrix operations in CKKS.

\begin{table}
  \caption{Performance of plaintext frequency-change detection.}
  \label{tab:plaintext-cpd-perf}
  \begin{tabular}{rrr}
    \toprule
    Length & Runtime (s) & Memory (MB) \\
    \midrule
    1,000 & 0.013 & 0.026 \\
    10,000 & 0.018 & 0.292 \\
    100,000 & 0.152 & 2.491 \\
    1,000,000 & 1.551 & 24.467 \\
  \bottomrule
\end{tabular}
\end{table}

The runtime provided in Figure~\ref{fig:runtime} refers to the total processing time of the encrypted time series on the server-side.
For a runtime breakdown and memory cost, we refer to Table~\ref{tab:runtime-micro}, where we also include the encryption costs on the client-side.
We note that the client's encryption time and ciphertext size scale linearly with the number of data points, with the time ranging from 0.5-0.9s for 1,000 points to 10.7-17.5s for 1,000,000 points, and the ciphertext size going up to 2.2GB.
In practice, the encrypted data would be streamed in \mytilde0.1GB ciphertext chunks, without waiting for all points to be collected.
Table~\ref{tab:plaintext-cpd-perf} reports the runtime and memory consumption of the plaintext baseline for frequency-change detection, which shows that the homomorphic-encryption layer adds between two and three orders of magnitude of runtime overhead and between three and six orders of magnitude of memory overhead to the underlying algorithm.

\paragraph{Ablation Study}
To assess how the individual components of our approach improve over their naive implementations or existing solutions, we consider a time series with 10,000 points over which we want to perform frequency-shift detection.
These points are encrypted into one ciphertext as 100 blocks of 100 elements each.
Starting with the turning rates computation, our Algorithm~\ref{alg:enc-turning-rates} performs 1 SIMD comparison, $\ceil{\log_2{100}} + 2 = 9$ homomorphic rotations, and 2 homomorphic multiplications, leading to a runtime of 40.11s, with the comparison being the bottleneck of the computation (>90\% of the cost).
On the other hand, the naive approach described in Section~\ref{sec:summ_functions} would extract each triplet of each block and process it separately, requiring 9800 SIMD comparisons (one for each triplet ranking), in addition to other relatively smaller costs, for an estimated runtime of \mytilde100.45h.
Improving on the naive strategy, one could pack multiple ranking computations in the same ciphertext, requiring in this case 2 SIMD comparisons, $4 \ceil{\log_2{100}} + 4 = 32$ rotations, and 3 multiplications, for a cost of 86.85s (thus, a 2.17× speed-up for our design).
The actual cost to integrate this method in the overall CPD pipeline would be higher, as the output encoding would not match what we need for computing the CUSUM statistic (see Section~\ref{sec:secure-cpd:assumptions-and-matrix-encoding}).

As for the partial sums, our method (line 2 of Algorithm~\ref{alg:enc-cusum}) requires $2 \ceil{\log_2{100}} = 14$ rotations and 1 multiplication, leading to a runtime of 8.63s.
On the other hand, the naive solution described in Section~\ref{sec:secureCPD:cusum} would require $99$ multiplications to mask out each non-trivial sub-vector, $\ceil{\log_2{100}}$ rotations for each partial sum, hence 700 in total, resulting in a runtime of \mytilde 460s (thus, a 53.3× speed-up for our design).
Lastly, as already discussed in Section~\ref{sec:secureCPD:cusum}, our argmax design improves over the one by Mazzone et al.~\cite{mazzone2025efficient}, more than halving the number of homomorphic multiplications (from 76 to 33), and leading to a runtime reduction from 54.16s to 23.54s (2.30× speed-up).

\begin{table*}[t]
  \caption{Accuracy of our approach on synthetic data with $n = 40{,}000$ on different distributions.}
  \label{tab:synth-accuracy}
  \begin{tabular}{llrrr}
    \toprule
    Change Type & Distributional change & Ground-Truth & Our Approach (rel. error) & Plaintext Baseline (rel. error) \\
    \midrule
    mean & $\mathcal{N}(0,1)$ to $\mathcal{N}(1,1)$ & 20,000 & 20,000 (0.00\%) & 20,000 (0.00\%) \\
    mean & $\mathcal{U}(0,1)$ to $\mathcal{U}(1,1)$ & 20,000 & 20,000 (0.00\%) & 20,000 (0.00\%) \\
    variance & $\mathcal{N}(0,1)$ to $\mathcal{N}(0,2)$ & 20,000 & 20,000 (0.00\%) & 20,000 (0.00\%) \\
    variance & $\mathcal{U}(0,1)$ to $\mathcal{U}(0,2)$ & 20,000 & 20,000 (0.00\%) & 20,000 (0.00\%) \\
    frequency & $\mathcal{N}(0,1)$ with $\phi = 0.3 \rightarrow 0.7$ & 20,000 & 20,000 (0.00\%) & 20,000 (0.00\%) \\
    frequency & $\text{Lap}(0,4)$ with $\phi = 0.3 \rightarrow 0.7$ & 20,000 & 20,000 (0.00\%) & 20,000 (0.00\%) \\
    frequency & $t_5$ with $\phi = 0.3 \rightarrow 0.7$ & 20,000 & 20,000 (0.00\%) & 20,000 (0.00\%) \\
  \bottomrule
\end{tabular}
\end{table*}

\paragraph{Accuracy}
We evaluate the accuracy of our method against both the ground truth and the plaintext baseline.
The detection error is measured using the \emph{relative error}, defined as the normalized deviation of the estimated change point $\hat{\tau}$ from the true change point $\tau$:
\begin{equation}
\label{eq:error}
\text{error} := \left| \hat{\tau} - \tau \right| / \tau.
\end{equation}
On synthetic data, our approach matches the plaintext baseline exactly across different types of distributional changes, as shown in Table~\ref{tab:synth-accuracy} for $n=40{,}000$.
For short synthetic time series, a small relative error is introduced due to the limited number of samples.
However, both the plaintext baseline and our approach produce identical estimates and therefore the same relative error.
This behavior is a limitation of the statistical method and not of the encrypted computation.
When applied to real-world data, our approach achieves accuracy comparable to the plaintext baseline, with only a small deviation from the ground truth (Table~\ref{tab:real-world-perf2}).
Overall, the relative error remains minimal, typically within one percentage point, and may partly reflect uncertainty in the human-annotated ground-truth labels.
Figure~\ref{fig:CP-datasets} illustrates representative real-world time series together with the corresponding turning rate, CUSUM statistic, and detected change point.

\begin{table*}[t]
  \caption{Performance of our approach on real-world datasets for detecting frequency changes.}
  \label{tab:real-world-perf2}
  \begin{tabular}{lrrrr}
    \toprule
    Dataset & Runtime (s) & Ground-Truth & Our Approach (rel. error) & Plaintext Baseline (rel. error) \\
    \midrule
    EEG & 91.19 & 46,500 & 46,360 (0.30\%) & 46,604 (0.22\%) \\
    Meditation & 78.93 & 4,049 & 4,005 (1.09\%) & 4,005 (1.09\%) \\
    Network & 89.19 & 23,752 & 23,800 (0.20\%) & 23,800 (0.20\%) \\
  \bottomrule
\end{tabular}
\end{table*}

\begin{figure}[!t]
   \begin{tikzpicture}
  \pgfplotsset{scaled y ticks=false}
  \small
  \begin{axis}[
    xlabel={Privacy budget ($\varepsilon$)},
    ylabel={Relative error},
    grid=both,
    major grid style={line width=.2pt, draw=gray!30},
    minor grid style={line width=.1pt, draw=gray!10},
    enlarge x limits=0.02,
    enlarge y limits=0.05,
    ymin=0,
    height=6cm,
    width=\linewidth,
    legend columns=2,
    legend cell align={left},
    legend style={
        at={(0.5,1.05)},
        /tikz/every even column/.append style={column sep=0.2cm},
        anchor=south,
        draw=none, 
        yshift=-2pt
    }
  ]

\addplot[name path=dp1_low, draw=none, forget plot] coordinates {
(0.5,0.365720904)
(1,0.365316514)
(1.5,0.364669824)
(2,0.364231363)
(2.5,0.36433796)
(3,0.364891395)
(3.5,0.364679435)
(4,0.364707474)
(4.5,0.364636508)
(5,0.36387812)
(5.5,0.363835587)
(6,0.363865723)
(6.5,0.364147077)
(7,0.364728258)
(7.5,0.364536048)
(8,0.364794427)
(8.5,0.364355932)
(9,0.364372016)
(9.5,0.364506506)
(10,0.36482664)
(10.5,0.364435584)
(11,0.364465498)
(11.5,0.364585114)
(12,0.363588531)
(12.5,0.363712545)
(13,0.363837778)
(13.5,0.363968889)
(14,0.363669576)
(14.5,0.362901751)
(15,0.363462044)
(15.5,0.363483496)
(16,0.363138021)
(16.5,0.362952089)
(17,0.362852519)
(17.5,0.361848222)
(18,0.361362409)
(18.5,0.36204961)
(19,0.362050263)
(19.5,0.360614962)
(20,0.361007409)
(20.5,0.360814701)
(21,0.359844942)
(21.5,0.359105506)
(22,0.358429161)
(22.5,0.358292636)
(23,0.35769298)
(23.5,0.356999297)
(24,0.356376202)
(24.5,0.355579763)
(25,0.355150919)
(25.5,0.355055788)
(26,0.354722597)
(26.5,0.354549532)
(27,0.353871137)
(27.5,0.352664933)
(28,0.351834154)
(28.5,0.351805763)
(29,0.350729364)
(29.5,0.350239418)
(30,0.34948638)
};
\addplot[name path=dp1_high, draw=none, forget plot] coordinates {
(0.5,0.368649176)
(1,0.368240406)
(1.5,0.367592816)
(2,0.367153877)
(2.5,0.36725676)
(3,0.367809365)
(3.5,0.367597165)
(4,0.367629446)
(4.5,0.367560292)
(5,0.36680056)
(5.5,0.366759213)
(6,0.366789557)
(6.5,0.367072083)
(7,0.367656222)
(7.5,0.367467392)
(8,0.367724893)
(8.5,0.367289188)
(9,0.367305144)
(9.5,0.367443174)
(10,0.36776032)
(10.5,0.367371696)
(11,0.367401502)
(11.5,0.367517926)
(12,0.366519309)
(12.5,0.366647135)
(13,0.366772302)
(13.5,0.366901751)
(14,0.366603104)
(14.5,0.365833569)
(15,0.366391356)
(15.5,0.366414424)
(16,0.366070259)
(16.5,0.365882871)
(17,0.365786161)
(17.5,0.364781258)
(18,0.364293991)
(18.5,0.36498303)
(19,0.364987577)
(19.5,0.363550278)
(20,0.363941751)
(20.5,0.363750659)
(21,0.362776818)
(21.5,0.362038134)
(22,0.361361959)
(22.5,0.361225564)
(23,0.36062534)
(23.5,0.359933383)
(24,0.359310158)
(24.5,0.358510797)
(25,0.358080921)
(25.5,0.357987692)
(26,0.357655563)
(26.5,0.357487148)
(27,0.356809823)
(27.5,0.355599467)
(28,0.354770246)
(28.5,0.354742437)
(29,0.353663876)
(29.5,0.353178862)
(30,0.35241962)
};

\addplot[blue!70!white, fill opacity=0.20, draw=none, forget plot] fill between[of=dp1_low and dp1_high];

\addplot[blue!70!white, mark=*,mark size=1.0,thick] coordinates {
(0.5,0.36718504)
(1,0.36677846)
(1.5,0.36613132)
(2,0.36569262)
(2.5,0.36579736)
(3,0.36635038)
(3.5,0.3661383)
(4,0.36616846)
(4.5,0.3660984)
(5,0.36533934)
(5.5,0.3652974)
(6,0.36532764)
(6.5,0.36560958)
(7,0.36619224)
(7.5,0.36600172)
(8,0.36625966)
(8.5,0.36582256)
(9,0.36583858)
(9.5,0.36597484)
(10,0.36629348)
(10.5,0.36590364)
(11,0.3659335)
(11.5,0.36605152)
(12,0.36505392)
(12.5,0.36517984)
(13,0.36530504)
(13.5,0.36543532)
(14,0.36513634)
(14.5,0.36436766)
(15,0.3649267)
(15.5,0.36494896)
(16,0.36460414)
(16.5,0.36441748)
(17,0.36431934)
(17.5,0.36331474)
(18,0.3628282)
(18.5,0.36351632)
(19,0.36351892)
(19.5,0.36208262)
(20,0.36247458)
(20.5,0.36228268)
(21,0.36131088)
(21.5,0.36057182)
(22,0.35989556)
(22.5,0.3597591)
(23,0.35915916)
(23.5,0.35846634)
(24,0.35784318)
(24.5,0.35704528)
(25,0.35661592)
(25.5,0.35652174)
(26,0.35618908)
(26.5,0.35601834)
(27,0.35534048)
(27.5,0.3541322)
(28,0.3533022)
(28.5,0.3532741)
(29,0.35219662)
(29.5,0.35170914)
(30,0.350953)
};
\addlegendentry{$n=\text{1,000}, N=\text{100,000}$}
\addplot[solid, blue!70!white] coordinates {(0, 0.17587362) (30, 0.17587362)};
\addlegendentry{baseline $n=\text{1,000}$}

\addplot[name path=dp2_low, draw=none, forget plot] coordinates {
(0.5,0.364036013)
(1,0.366562927)
(1.5,0.367046922)
(2,0.365192956)
(2.5,0.368223796)
(3,0.370457356)
(3.5,0.368490008)
(4,0.367049993)
(4.5,0.36963189)
(5,0.371553937)
(5.5,0.372341962)
(6,0.369240541)
(6.5,0.371122246)
(7,0.369756696)
(7.5,0.368241914)
(8,0.369596074)
(8.5,0.369002193)
(9,0.370239345)
(9.5,0.366913757)
(10,0.365364109)
(10.5,0.36703279)
(11,0.367853893)
(11.5,0.364492673)
(12,0.365662398)
(12.5,0.365935692)
(13,0.363829811)
(13.5,0.36224187)
(14,0.360659742)
(14.5,0.35685897)
(15,0.359198556)
(15.5,0.355641207)
(16,0.356872175)
(16.5,0.356519919)
(17,0.356124276)
(17.5,0.349688516)
(18,0.3502158)
(18.5,0.34455303)
(19,0.343227107)
(19.5,0.339895268)
(20,0.336444716)
(20.5,0.337021913)
(21,0.333429343)
(21.5,0.330650154)
(22,0.32830702)
(22.5,0.325132309)
(23,0.321940817)
(23.5,0.319690104)
(24,0.315468308)
(24.5,0.310459232)
(25,0.306482944)
(25.5,0.302420896)
(26,0.297360871)
(26.5,0.293825051)
(27,0.290430874)
(27.5,0.285313658)
(28,0.279704395)
(28.5,0.277635459)
(29,0.27442758)
(29.5,0.272107717)
(30,0.268337915)
};
\addplot[name path=dp2_high, draw=none, forget plot] coordinates {
(0.5,0.373355987)
(1,0.375833073)
(1.5,0.376285078)
(2,0.374475044)
(2.5,0.377544204)
(3,0.379754644)
(3.5,0.377789992)
(4,0.376370007)
(4.5,0.37898011)
(5,0.380886063)
(5.5,0.381718038)
(6,0.378571459)
(6.5,0.380513754)
(7,0.379095304)
(7.5,0.377602086)
(8,0.378935926)
(8.5,0.378333807)
(9,0.379584655)
(9.5,0.376242243)
(10,0.374635891)
(10.5,0.37627921)
(11,0.377138107)
(11.5,0.373731327)
(12,0.374925602)
(12.5,0.375180308)
(13,0.373122189)
(13.5,0.37154613)
(14,0.369940258)
(14.5,0.36612503)
(15,0.368493444)
(15.5,0.364938793)
(16,0.366187825)
(16.5,0.365896081)
(17,0.365503724)
(17.5,0.359003484)
(18,0.3595882)
(18.5,0.35390297)
(19,0.352616893)
(19.5,0.349256732)
(20,0.345779284)
(20.5,0.346354087)
(21,0.342734657)
(21.5,0.339985846)
(22,0.33763298)
(22.5,0.334491691)
(23,0.331291183)
(23.5,0.329041896)
(24,0.324803692)
(24.5,0.319708768)
(25,0.315765056)
(25.5,0.311719104)
(26,0.306567129)
(26.5,0.302978949)
(27,0.299557126)
(27.5,0.294370342)
(28,0.288719605)
(28.5,0.286668541)
(29,0.28344842)
(29.5,0.281136283)
(30,0.277318085)
};

\addplot[orange, fill opacity=0.20, draw=none, forget plot] fill between[of=dp2_low and dp2_high];

\addplot[orange, mark=square*,mark size=1.0,thick] coordinates {
(0.5,0.368696)
(1,0.371198)
(1.5,0.371666)
(2,0.369834)
(2.5,0.372884)
(3,0.375106)
(3.5,0.37314)
(4,0.37171)
(4.5,0.374306)
(5,0.37622)
(5.5,0.37703)
(6,0.373906)
(6.5,0.375818)
(7,0.374426)
(7.5,0.372922)
(8,0.374266)
(8.5,0.373668)
(9,0.374912)
(9.5,0.371578)
(10,0.37)
(10.5,0.371656)
(11,0.372496)
(11.5,0.369112)
(12,0.370294)
(12.5,0.370558)
(13,0.368476)
(13.5,0.366894)
(14,0.3653)
(14.5,0.361492)
(15,0.363846)
(15.5,0.36029)
(16,0.36153)
(16.5,0.361208)
(17,0.360814)
(17.5,0.354346)
(18,0.354902)
(18.5,0.349228)
(19,0.347922)
(19.5,0.344576)
(20,0.341112)
(20.5,0.341688)
(21,0.338082)
(21.5,0.335318)
(22,0.33297)
(22.5,0.329812)
(23,0.326616)
(23.5,0.324366)
(24,0.320136)
(24.5,0.315084)
(25,0.311124)
(25.5,0.30707)
(26,0.301964)
(26.5,0.298402)
(27,0.294994)
(27.5,0.289842)
(28,0.284212)
(28.5,0.282152)
(29,0.278938)
(29.5,0.276622)
(30,0.272828)
};
\addlegendentry{$n=\text{10,000}, N=\text{10,000}$}
\addplot[dashed, orange] coordinates {(0, 0.024426) (30, 0.024426)};
\addlegendentry{baseline $n=\text{10,000}$}

\addplot[name path=dp3_low, draw=none, forget plot] coordinates {
(0.5,0.352167334)
(1,0.353371587)
(1.5,0.351884457)
(2,0.355639738)
(2.5,0.339159568)
(3,0.352662522)
(3.5,0.348066468)
(4,0.346368837)
(4.5,0.352208167)
(5,0.347726951)
(5.5,0.354942482)
(6,0.354458103)
(6.5,0.362058475)
(7,0.35429069)
(7.5,0.347537991)
(8,0.345002821)
(8.5,0.341458992)
(9,0.344361029)
(9.5,0.346710157)
(10,0.339552787)
(10.5,0.331725205)
(11,0.325210833)
(11.5,0.325471845)
(12,0.31925967)
(12.5,0.313105038)
(13,0.31145599)
(13.5,0.304053644)
(14,0.30299744)
(14.5,0.293814105)
(15,0.28772429)
(15.5,0.290209608)
(16,0.279462821)
(16.5,0.267919233)
(17,0.255498115)
(17.5,0.250333095)
(18,0.235445471)
(18.5,0.23504996)
(19,0.225714974)
(19.5,0.216412946)
(20,0.212568204)
(20.5,0.203751397)
(21,0.194398496)
(21.5,0.186370792)
(22,0.17701893)
(22.5,0.169049529)
(23,0.160107807)
(23.5,0.150963538)
(24,0.144042253)
(24.5,0.132817081)
(25,0.124554409)
(25.5,0.119082386)
(26,0.113175365)
(26.5,0.109779185)
(27,0.104905455)
(27.5,0.101768047)
(28,0.096026769)
(28.5,0.091467612)
(29,0.088709767)
(29.5,0.083525058)
(30,0.079223355)
};
\addplot[name path=dp3_high, draw=none, forget plot] coordinates {
(0.5,0.382151226)
(1,0.383848253)
(1.5,0.382098103)
(2,0.385770982)
(2.5,0.368789712)
(3,0.382514598)
(3.5,0.377696092)
(4,0.375615483)
(4.5,0.381607033)
(5,0.376729209)
(5.5,0.384888558)
(6,0.383522697)
(6.5,0.391199445)
(7,0.38368323)
(7.5,0.376858329)
(8,0.374221339)
(8.5,0.370769808)
(9,0.373314331)
(9.5,0.375917843)
(10,0.368572493)
(10.5,0.359974795)
(11,0.354043407)
(11.5,0.354552155)
(12,0.34828289)
(12.5,0.342086962)
(13,0.34028305)
(13.5,0.332725396)
(14,0.33179632)
(14.5,0.322802055)
(15,0.31672211)
(15.5,0.319492472)
(16,0.308225979)
(16.5,0.296487807)
(17,0.283670365)
(17.5,0.278592025)
(18,0.263233249)
(18.5,0.26258684)
(19,0.252896226)
(19.5,0.243461294)
(20,0.239427796)
(20.5,0.230183323)
(21,0.219974144)
(21.5,0.211506808)
(22,0.20094923)
(22.5,0.192847431)
(23,0.182672833)
(23.5,0.172906542)
(24,0.165017427)
(24.5,0.152738919)
(25,0.143592951)
(25.5,0.137224814)
(26,0.130838875)
(26.5,0.127449775)
(27,0.121766705)
(27.5,0.118090033)
(28,0.110720111)
(28.5,0.105865988)
(29,0.102606073)
(29.5,0.097023742)
(30,0.092209925)
};

\addplot[green!60!black, fill opacity=0.20, draw=none, forget plot] fill between[of=dp3_low and dp3_high];

\addplot[green!60!black, mark=triangle*,mark size=1.2,thick] coordinates {
(0.5,0.36715928)
(1,0.36860992)
(1.5,0.36699128)
(2,0.37070536)
(2.5,0.35397464)
(3,0.36758856)
(3.5,0.36288128)
(4,0.36099216)
(4.5,0.3669076)
(5,0.36222808)
(5.5,0.36991552)
(6,0.3689904)
(6.5,0.37662896)
(7,0.36898696)
(7.5,0.36219816)
(8,0.35961208)
(8.5,0.3561144)
(9,0.35883768)
(9.5,0.361314)
(10,0.35406264)
(10.5,0.34585)
(11,0.33962712)
(11.5,0.340012)
(12,0.33377128)
(12.5,0.327596)
(13,0.32586952)
(13.5,0.31838952)
(14,0.31739688)
(14.5,0.30830808)
(15,0.3022232)
(15.5,0.30485104)
(16,0.2938444)
(16.5,0.28220352)
(17,0.26958424)
(17.5,0.26446256)
(18,0.24933936)
(18.5,0.2488184)
(19,0.2393056)
(19.5,0.22993712)
(20,0.225998)
(20.5,0.21696736)
(21,0.20718632)
(21.5,0.1989388)
(22,0.18898408)
(22.5,0.18094848)
(23,0.17139032)
(23.5,0.16193504)
(24,0.15452984)
(24.5,0.142778)
(25,0.13407368)
(25.5,0.1281536)
(26,0.12200712)
(26.5,0.11861448)
(27,0.11333608)
(27.5,0.10992904)
(28,0.10337344)
(28.5,0.0986668)
(29,0.09565792)
(29.5,0.0902744)
(30,0.08571664)
};
\addlegendentry{$n=\text{100,000}, N=\text{1,000}$}
\addplot[dotted, green!60!black] coordinates {(0, 0.00360119999999999) (30, 0.00360119999999999)};
\addlegendentry{baseline $n=\text{100,000}$}

\addplot[name path=dp4_low, draw=none, forget plot] coordinates {
(0.5,0.294754404)
(1,0.303760599)
(1.5,0.319606991)
(2,0.302132068)
(2.5,0.30503807)
(3,0.30242718)
(3.5,0.297107067)
(4,0.284617491)
(4.5,0.278602548)
(5,0.270316157)
(5.5,0.282347523)
(6,0.282186509)
(6.5,0.294526277)
(7,0.283071403)
(7.5,0.268476962)
(8,0.270578401)
(8.5,0.25858549)
(9,0.261308006)
(9.5,0.256697323)
(10,0.263083748)
(10.5,0.24001361)
(11,0.22060903)
(11.5,0.219647324)
(12,0.199536925)
(12.5,0.195356795)
(13,0.165087026)
(13.5,0.166614439)
(14,0.151339002)
(14.5,0.144795405)
(15,0.122401047)
(15.5,0.109433013)
(16,0.105095401)
(16.5,0.089314331)
(17,0.087304155)
(17.5,0.07687388)
(18,0.064800339)
(18.5,0.059459056)
(19,0.052950402)
(19.5,0.051216407)
(20,0.040870653)
(20.5,0.038829922)
(21,0.038455199)
(21.5,0.035408174)
(22,0.034926152)
(22.5,0.033964765)
(23,0.02925435)
(23.5,0.023597073)
(24,0.021309753)
(24.5,0.017093779)
(25,0.01431888)
(25.5,0.014217409)
(26,0.01199645)
(26.5,0.010929447)
(27,0.011205837)
(27.5,0.010105281)
(28,0.010015268)
(28.5,0.010460457)
(29,0.009323213)
(29.5,0.008726179)
(30,0.0078077)
};
\addplot[name path=dp4_high, draw=none, forget plot] coordinates {
(0.5,0.381005596)
(1,0.387759401)
(1.5,0.401993009)
(2,0.380987932)
(2.5,0.39704193)
(3,0.38649282)
(3.5,0.387492933)
(4,0.374022509)
(4.5,0.364077452)
(5,0.356083843)
(5.5,0.372532477)
(6,0.372053491)
(6.5,0.382633723)
(7,0.361448597)
(7.5,0.344323038)
(8,0.352341599)
(8.5,0.34141451)
(9,0.344171994)
(9.5,0.347382677)
(10,0.349196252)
(10.5,0.32006639)
(11,0.30071097)
(11.5,0.300112676)
(12,0.280063075)
(12.5,0.278363205)
(13,0.240912974)
(13.5,0.244825561)
(14,0.226540998)
(14.5,0.218044595)
(15,0.183758953)
(15.5,0.166126987)
(16,0.160344599)
(16.5,0.135885669)
(17,0.134775845)
(17.5,0.12128612)
(18,0.103359661)
(18.5,0.095220944)
(19,0.086049598)
(19.5,0.078823593)
(20,0.064929347)
(20.5,0.062250078)
(21,0.062184801)
(21.5,0.058791826)
(22,0.058753848)
(22.5,0.057915235)
(23,0.04930565)
(23.5,0.040842927)
(24,0.037090247)
(24.5,0.030266221)
(25,0.02552112)
(25.5,0.025182591)
(26,0.02184355)
(26.5,0.019750553)
(27,0.020234163)
(27.5,0.018974719)
(28,0.018104732)
(28.5,0.018339543)
(29,0.017356787)
(29.5,0.016633821)
(30,0.0155923)
};

\addplot[red, fill opacity=0.20, draw=none, forget plot] fill between[of=dp4_low and dp4_high];

\addplot[red,mark=x,mark size=1.5,thick] coordinates {
(0.5,0.33788)
(1,0.34576)
(1.5,0.3608)
(2,0.34156)
(2.5,0.35104)
(3,0.34446)
(3.5,0.3423)
(4,0.32932)
(4.5,0.32134)
(5,0.3132)
(5.5,0.32744)
(6,0.32712)
(6.5,0.33858)
(7,0.32226)
(7.5,0.3064)
(8,0.31146)
(8.5,0.3)
(9,0.30274)
(9.5,0.30204)
(10,0.30614)
(10.5,0.28004)
(11,0.26066)
(11.5,0.25988)
(12,0.2398)
(12.5,0.23686)
(13,0.203)
(13.5,0.20572)
(14,0.18894)
(14.5,0.18142)
(15,0.15308)
(15.5,0.13778)
(16,0.13272)
(16.5,0.1126)
(17,0.11104)
(17.5,0.09908)
(18,0.08408)
(18.5,0.07734)
(19,0.0695)
(19.5,0.06502)
(20,0.0529)
(20.5,0.05054)
(21,0.05032)
(21.5,0.0471)
(22,0.04684)
(22.5,0.04594)
(23,0.03928)
(23.5,0.03222)
(24,0.0292)
(24.5,0.02368)
(25,0.01992)
(25.5,0.0197)
(26,0.01692)
(26.5,0.01534)
(27,0.01572)
(27.5,0.01454)
(28,0.01406)
(28.5,0.0144)
(29,0.01334)
(29.5,0.01268)
(30,0.0117)
};
\addlegendentry{$n=\text{1,000,000}, N=\text{100}$}
\addplot[dashdotted, red] coordinates {(0, 0.00004) (30, 0.00004)};
\addlegendentry{baseline $n=\text{1,000,000}$}

  \end{axis}
\end{tikzpicture}
\caption{Relative error of the local-DP approach on four AR(1) series. The shaded areas show $\pm 2\,\mathrm{SEM}$ around the mean relative error.
The horizontal dashed lines indicate the relative error of our approach as a baseline.}
    \label{fig:local_DP}
\end{figure}
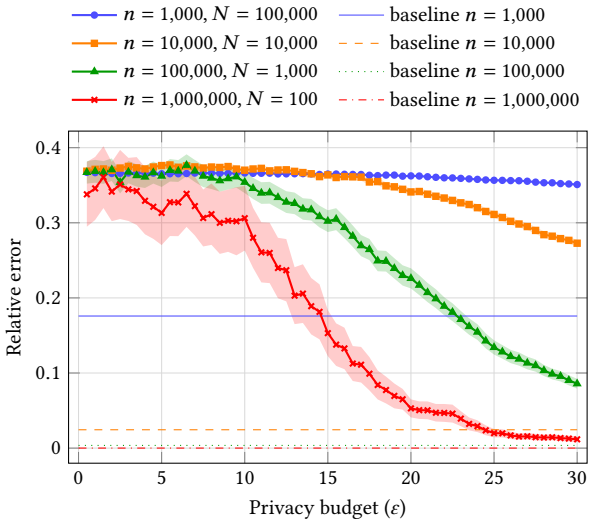

\begin{figure}
\centering
\newcommand{\subfigheight}{60pt}
\begin{subfigure}[b]{\columnwidth}
    \centering
    \includegraphics[width=\linewidth]{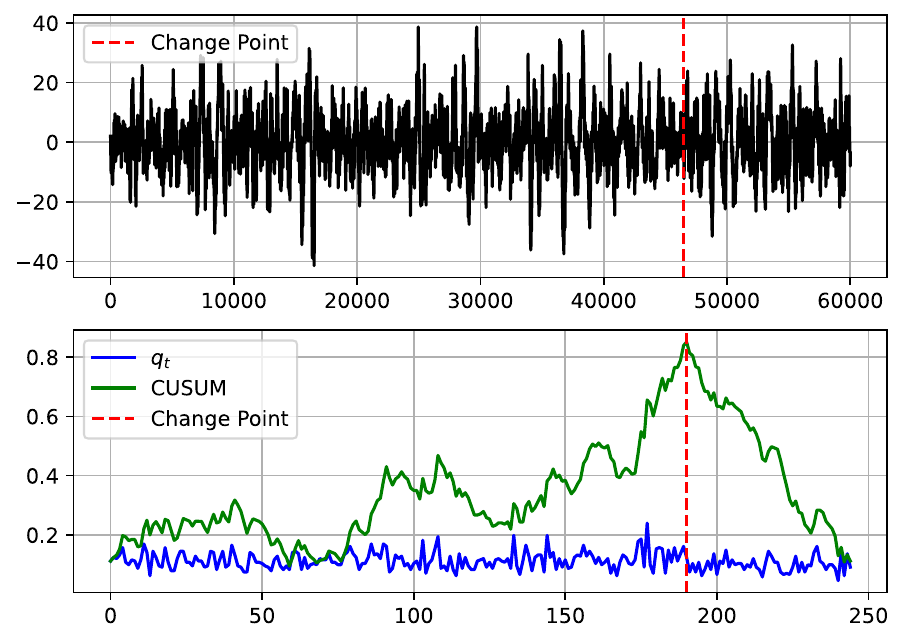}
    \caption{Sleep Phase Change Detection in EEG}
\end{subfigure}
\vspace{0pt}

\begin{subfigure}[b]{\columnwidth}
    \centering
    \includegraphics[width=\linewidth]{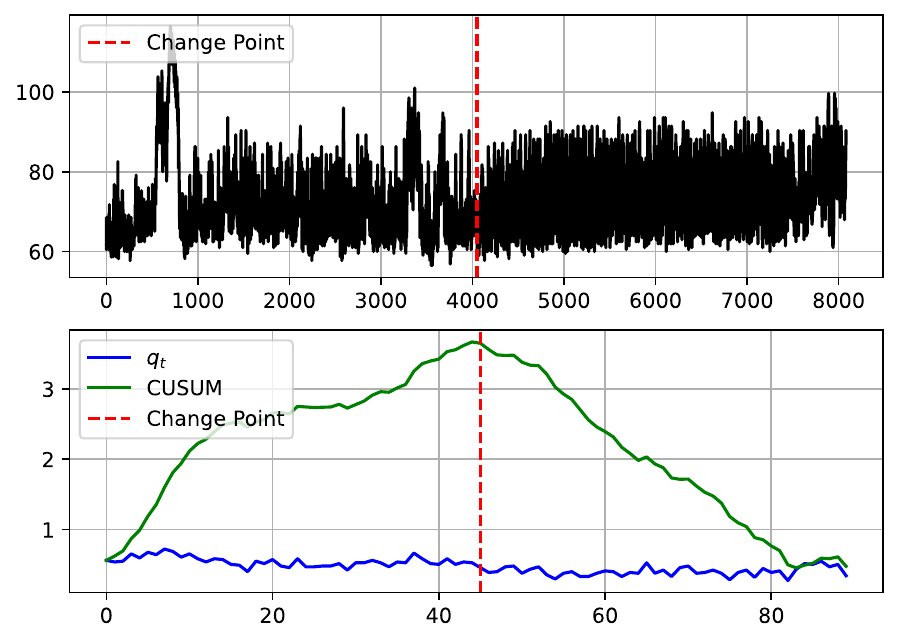}
    \caption{Heartbeat Interval Time Series}
\end{subfigure}
\vspace{0pt}

\begin{subfigure}[b]{\columnwidth}
    \centering
    \includegraphics[width=\linewidth]{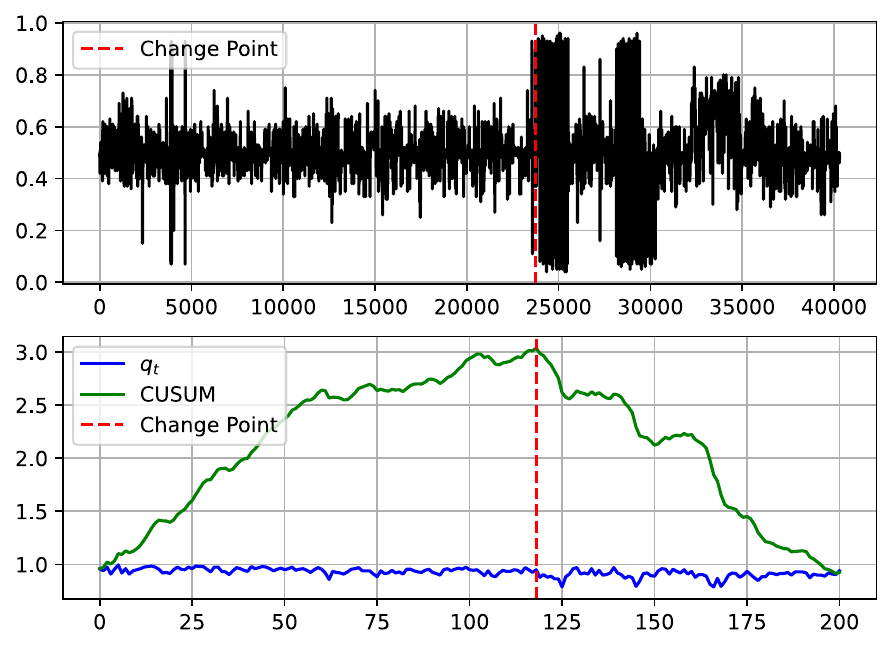}
    \caption{Network Traffic Anomaly Detection}
\end{subfigure}

\caption{Real-world time-series data with the detected change point (vertical dashed line).}
\label{fig:CP-datasets}
\end{figure}

\subsection{Comparison with DP-Based Solutions}
\label{sec:experiments_comparisonDP}

The only privacy-preserving CPD solutions available in the literature \cite{NEURIPS2018_f19ec2b8,JMLR:v22:19-770} work in the \emph{central} DP model, where the raw data are accessible to a trusted curator that applies a randomized algorithm to produce a privacy-preserving output.
In contrast, our solution operates entirely on encrypted data without requiring access to plaintext or any trusted server.
This is closer to the \textit{local} DP model, where the data curator applies noise to the data before outsourcing it to a (non-trusted) computing server.
These two paradigms are fundamentally different, both in terms of trust assumptions and privacy guarantees, making a direct comparison meaningless.
To compare our approach against DP, we opted for implementing a simple local-DP solution ourselves.

\subsubsection*{Comparison with Local-DP}
Our local-DP solution operates by adding calibrated noise to a clipped version of the time series before applying a standard CUSUM procedure.
The resulting privacy guarantees follow from the post-processing property of DP.
For a more detailed description, we refer to Appendix~\ref{appendix:DP}. 
We assess our approach against the local-DP solution on the AR(1) synthetic data, and report the results in Figure~\ref{fig:local_DP}.
The parameter $\delta$ is set to $1/n^2$ as commonly done in DP applications, and the error is computed via Equation~\ref{eq:error} and averaged over a large number of experiments ($N$) to account for its high variability.
For the clipping bound, we set $M = 1$, which empirically provides a reasonable trade-off between input information loss due to clipping and an overly large sensitivity.
The runtime and memory performance of the local-DP solution are of the same order of magnitude as the plaintext baseline, with client runtime being 0.1-9.3ms, server runtime going from 3.7ms to 1.54s, and memory usage upper bounded by \mytilde 40MB on both sides.
However, the accuracy degrades significantly as \(\varepsilon\) decreases, with relative error going below $20\%$ only for $\varepsilon > 13$ (for 1,000,000 points), making local-DP impractical in scenarios where accuracy is required.

\subsubsection*{Central DP}
We also include an empirical evaluation of \cite{NEURIPS2018_f19ec2b8} and \cite{JMLR:v22:19-770}.
These experiments, detailed in Appendix~\ref{app:central_DP}, are meant to provide an indicative sense of the accuracy trade-offs under central DP guarantees and are only presented for completeness.

It is important to note that both these algorithms assume that the data consists of independent observations and that the distributions before and after the change point are known in advance. For this reason, evaluations are performed on synthetic data designed to satisfy these assumptions. As remarked throughout our work, in real-world applications, time series exhibit temporal dependence, and pre- and post-change distributions are typically unknown and must be estimated.
Even when independence approximately holds, estimating the distributions introduces additional statistical noise and model misspecification.

It is also worth noting that the local DP mechanism we consider introduces noise proportional to \(1/\sqrt{\varepsilon}\), which is substantially larger than the \(1/\varepsilon\) noise scale required under central DP mechanisms. This discrepancy arises because, in our setting, privacy is enforced directly on the raw time series (effectively privatizing the identity function), and noise must be added to each individual data point. The total privacy noise thus accumulates over time, degrading the overall utility. In contrast, central DP schemes typically add noise only to aggregate statistics (e.g., means or sums), allowing for much lower noise magnitudes while still satisfying the same \((\varepsilon, \delta)\)-DP guarantees.

\section{Conclusion}
\label{sec:conclusion}

In this paper, we presented the first fully encrypted pipeline for change-point detection under homomorphic encryption.
Our solution combines ordinal-pattern-based estimators with recent advances in encrypted comparison, enabling the detection of structural changes in sensitive time series with plaintext-level accuracy and without exposing the underlying data. This establishes, for the first time, the feasibility of CPD in cryptographically private settings, a direction that had remained unexplored despite the central role of CPD in statistics.

As future work, our framework points to a general methodology for outsourced statistical analysis of time series via ordinal patterns.
In this direction, it will be possible to incorporate more refined statistics, such as permutation entropy or other complexity-based measures, into the encrypted domain.
A particularly challenging but promising extension is the treatment of multiple change points, which remains difficult even in the plaintext setting.

Finally, this work assumes operation under the alternative hypothesis, i.e., that a change point is known to exist.
Extending our approach to support full hypothesis testing under encryption, including the null hypothesis of no change, would require securely estimating decision thresholds or $p$-values, and represents an important open problem.
Addressing this challenge would further broaden the applicability of encrypted CPD to real-world scenarios where the presence of structural changes is not guaranteed.

\begin{acks}

We would like to thank Florian Hahn and Luca Mariot for their valuable insights.

\end{acks}

\bibliographystyle{ACM-Reference-Format}
\bibliography{main}


\begin{thebibliography}{38}


\ifx \showCODEN    \undefined \def \showCODEN     #1{\unskip}     \fi
\ifx \showISBNx    \undefined \def \showISBNx     #1{\unskip}     \fi
\ifx \showISBNxiii \undefined \def \showISBNxiii  #1{\unskip}     \fi
\ifx \showISSN     \undefined \def \showISSN      #1{\unskip}     \fi
\ifx \showLCCN     \undefined \def \showLCCN      #1{\unskip}     \fi
\ifx \shownote     \undefined \def \shownote      #1{#1}          \fi
\ifx \showarticletitle \undefined \def \showarticletitle #1{#1}   \fi
\ifx \showURL      \undefined \def \showURL       {\relax}        \fi
\providecommand\bibfield[2]{#2}
\providecommand\bibinfo[2]{#2}
\providecommand\natexlab[1]{#1}
\providecommand\showeprint[2][]{arXiv:#2}

\bibitem[Al~Badawi et~al\mbox{.}(2022)]%
        {openFHE}
\bibfield{author}{\bibinfo{person}{Ahmad Al~Badawi}, \bibinfo{person}{Jack Bates}, \bibinfo{person}{Flavio Bergamaschi}, \bibinfo{person}{David~Bruce Cousins}, \bibinfo{person}{Saroja Erabelli}, \bibinfo{person}{Nicholas Genise}, \bibinfo{person}{Shai Halevi}, \bibinfo{person}{Hamish Hunt}, \bibinfo{person}{Andrey Kim}, \bibinfo{person}{Yongwoo Lee}, \bibinfo{person}{Zeyu Liu}, \bibinfo{person}{Daniele Micciancio}, \bibinfo{person}{Ian Quah}, \bibinfo{person}{Yuriy Polyakov}, \bibinfo{person}{Saraswathy R.V.}, \bibinfo{person}{Kurt Rohloff}, \bibinfo{person}{Jonathan Saylor}, \bibinfo{person}{Dmitriy Suponitsky}, \bibinfo{person}{Matthew Triplett}, \bibinfo{person}{Vinod Vaikuntanathan}, {and} \bibinfo{person}{Vincent Zucca}.} \bibinfo{year}{2022}\natexlab{}.
\newblock \showarticletitle{OpenFHE: Open-Source Fully Homomorphic Encryption Library}. In \bibinfo{booktitle}{\emph{Proceedings of the 10th Workshop on Encrypted Computing \& Applied Homomorphic Cryptography}} (Los Angeles, CA, USA) \emph{(\bibinfo{series}{WAHC'22})}. \bibinfo{publisher}{Association for Computing Machinery}, \bibinfo{address}{New York, NY, USA}, \bibinfo{pages}{53–63}.
\newblock
\showISBNx{9781450398770}
\href{https://doi.org/10.1145/3560827.3563379}{doi:\nolinkurl{10.1145/3560827.3563379}}


\bibitem[Albrecht et~al\mbox{.}(2018)]%
        {HomomorphicEncryptionSecurityStandard}
\bibfield{author}{\bibinfo{person}{Martin Albrecht}, \bibinfo{person}{Melissa Chase}, \bibinfo{person}{Hao Chen}, \bibinfo{person}{Jintai Ding}, \bibinfo{person}{Shafi Goldwasser}, \bibinfo{person}{Sergey Gorbunov}, \bibinfo{person}{Shai Halevi}, \bibinfo{person}{Jeffrey Hoffstein}, \bibinfo{person}{Kim Laine}, \bibinfo{person}{Kristin Lauter}, \bibinfo{person}{Satya Lokam}, \bibinfo{person}{Daniele Micciancio}, \bibinfo{person}{Dustin Moody}, \bibinfo{person}{Travis Morrison}, \bibinfo{person}{Amit Sahai}, {and} \bibinfo{person}{Vinod Vaikuntanathan}.} \bibinfo{year}{2018}\natexlab{}.
\newblock \bibinfo{booktitle}{\emph{Homomorphic Encryption Security Standard}}.
\newblock \bibinfo{type}{{T}echnical {R}eport}. \bibinfo{institution}{HomomorphicEncryption.org}, \bibinfo{address}{Toronto, Canada}.
\newblock


\bibitem[Albrecht et~al\mbox{.}(2015)]%
        {albrecht2015concrete}
\bibfield{author}{\bibinfo{person}{Martin~R. Albrecht}, \bibinfo{person}{Rachel Player}, {and} \bibinfo{person}{Sam Scott}.} \bibinfo{year}{2015}\natexlab{}.
\newblock \showarticletitle{On the concrete hardness of Learning with Errors}.
\newblock \bibinfo{journal}{\emph{Journal of Mathematical Cryptology}} \bibinfo{volume}{9}, \bibinfo{number}{3} (\bibinfo{year}{2015}), \bibinfo{pages}{169--203}.
\newblock
\href{https://doi.org/doi:10.1515/jmc-2015-0016}{doi:\nolinkurl{doi:10.1515/jmc-2015-0016}}


\bibitem[Aston and Kirch(2012)]%
        {aston2012evaluating}
\bibfield{author}{\bibinfo{person}{John A.~D. Aston} {and} \bibinfo{person}{Claudia Kirch}.} \bibinfo{year}{2012}\natexlab{}.
\newblock \showarticletitle{{Evaluating stationarity via change-point alternatives with applications to fMRI data}}.
\newblock \bibinfo{journal}{\emph{The Annals of Applied Statistics}} \bibinfo{volume}{6}, \bibinfo{number}{4} (\bibinfo{year}{2012}), \bibinfo{pages}{1906 -- 1948}.
\newblock
\href{https://doi.org/10.1214/12-AOAS565}{doi:\nolinkurl{10.1214/12-AOAS565}}


\bibitem[Bai and Perron(1998)]%
        {bai1998estimating}
\bibfield{author}{\bibinfo{person}{Jushan Bai} {and} \bibinfo{person}{Pierre Perron}.} \bibinfo{year}{1998}\natexlab{}.
\newblock \showarticletitle{Estimating and Testing Linear Models with Multiple Structural Changes}.
\newblock \bibinfo{journal}{\emph{Econometrica}} \bibinfo{volume}{66}, \bibinfo{number}{1} (\bibinfo{year}{1998}), \bibinfo{pages}{47--78}.
\newblock
\href{https://doi.org/10.2307/2998540}{doi:\nolinkurl{10.2307/2998540}}


\bibitem[Bandt(2020)]%
        {bandt2020order}
\bibfield{author}{\bibinfo{person}{Christoph Bandt}.} \bibinfo{year}{2020}\natexlab{}.
\newblock \showarticletitle{Order patterns, their variation and change points in financial time series and {B}rownian motion}.
\newblock \bibinfo{journal}{\emph{Stat Papers}}  \bibinfo{volume}{61} (\bibinfo{year}{2020}), \bibinfo{pages}{1565--1588}.
\newblock
\href{https://doi.org/10.1007/s00362-020-01171-7}{doi:\nolinkurl{10.1007/s00362-020-01171-7}}


\bibitem[Bandt and Pompe(2002)]%
        {Bandt-Pompe}
\bibfield{author}{\bibinfo{person}{Christoph Bandt} {and} \bibinfo{person}{Bernd Pompe}.} \bibinfo{year}{2002}\natexlab{}.
\newblock \showarticletitle{Permutation Entropy: A Natural Complexity Measure for Time Series}.
\newblock \bibinfo{journal}{\emph{Phys. Rev. Lett.}}  \bibinfo{volume}{88} (\bibinfo{date}{Apr} \bibinfo{year}{2002}), \bibinfo{pages}{174102}.
\newblock
Issue 17.
\href{https://doi.org/10.1103/PhysRevLett.88.174102}{doi:\nolinkurl{10.1103/PhysRevLett.88.174102}}


\bibitem[Basseville and Nikiforov(1993)]%
        {basseville1993abrupt}
\bibfield{author}{\bibinfo{person}{Mich{\`e}le Basseville} {and} \bibinfo{person}{Igor~V. Nikiforov}.} \bibinfo{year}{1993}\natexlab{}.
\newblock \bibinfo{booktitle}{\emph{Detection of Abrupt Changes: Theory and Application}}.
\newblock \bibinfo{publisher}{Prentice Hall}, \bibinfo{address}{Englewood Cliffs, NJ}.
\newblock


\bibitem[Betken et~al\mbox{.}(2025)]%
        {betken2025ordinal}
\bibfield{author}{\bibinfo{person}{Annika Betken}, \bibinfo{person}{Giorgio Micali}, {and} \bibinfo{person}{Johannes Schmidt-Hieber}.} \bibinfo{year}{2025}\natexlab{}.
\newblock \showarticletitle{Ordinal pattern-based change point detection}.
\newblock \bibinfo{journal}{\emph{TEST}}  \bibinfo{volume}{34} (\bibinfo{year}{2025}), \bibinfo{pages}{927--980}.
\newblock
\href{https://doi.org/10.1007/s11749-025-00983-9}{doi:\nolinkurl{10.1007/s11749-025-00983-9}}


\bibitem[Blythe et~al\mbox{.}(2012)]%
        {Blythe2012}
\bibfield{author}{\bibinfo{person}{Duncan A.~J. Blythe}, \bibinfo{person}{Paul von Bunau}, \bibinfo{person}{Frank~C. Meinecke}, {and} \bibinfo{person}{Klaus-Robert Muller}.} \bibinfo{year}{2012}\natexlab{}.
\newblock \showarticletitle{Feature Extraction for Change-Point Detection Using Stationary Subspace Analysis}.
\newblock \bibinfo{journal}{\emph{IEEE Transactions on Neural Networks and Learning Systems}} \bibinfo{volume}{23}, \bibinfo{number}{4} (\bibinfo{year}{2012}), \bibinfo{pages}{631--643}.
\newblock
\href{https://doi.org/10.1109/TNNLS.2012.2185811}{doi:\nolinkurl{10.1109/TNNLS.2012.2185811}}


\bibitem[Braun et~al\mbox{.}(2000)]%
        {braun2000}
\bibfield{author}{\bibinfo{person}{J.~V. Braun}, \bibinfo{person}{R.~K. Braun}, {and} \bibinfo{person}{H.-G. M{\"u}ller}.} \bibinfo{year}{2000}\natexlab{}.
\newblock \showarticletitle{Multiple changepoint fitting via quasilikelihood, with application to DNA sequence segmentation}.
\newblock \bibinfo{journal}{\emph{Biometrika}} \bibinfo{volume}{87}, \bibinfo{number}{2} (\bibinfo{year}{2000}), \bibinfo{pages}{301--314}.
\newblock
\href{https://doi.org/10.1093/biomet/87.2.301}{doi:\nolinkurl{10.1093/biomet/87.2.301}}


\bibitem[Brockwell and Davis(1991)]%
        {brockwelldavis}
\bibfield{author}{\bibinfo{person}{Peter~J. Brockwell} {and} \bibinfo{person}{Richard~A. Davis}.} \bibinfo{year}{1991}\natexlab{}.
\newblock \bibinfo{booktitle}{\emph{Time Series: Theory and Methods} (\bibinfo{edition}{2nd} ed.)}.
\newblock \bibinfo{publisher}{Springer}.
\newblock
\showISBNx{978-0387974293}
\href{https://doi.org/10.1007/978-1-4419-0320-4}{doi:\nolinkurl{10.1007/978-1-4419-0320-4}}


\bibitem[Cheon et~al\mbox{.}(2017)]%
        {cheon2017homomorphic}
\bibfield{author}{\bibinfo{person}{Jung~Hee Cheon}, \bibinfo{person}{Andrey Kim}, \bibinfo{person}{Miran Kim}, {and} \bibinfo{person}{Yongsoo Song}.} \bibinfo{year}{2017}\natexlab{}.
\newblock \showarticletitle{Homomorphic Encryption for Arithmetic of Approximate Numbers}. In \bibinfo{booktitle}{\emph{Advances in Cryptology -- ASIACRYPT 2017}}, \bibfield{editor}{\bibinfo{person}{Tsuyoshi Takagi} {and} \bibinfo{person}{Thomas Peyrin}} (Eds.). \bibinfo{publisher}{Springer International Publishing}, \bibinfo{address}{Cham}, \bibinfo{pages}{409--437}.
\newblock
\showISBNx{978-3-319-70694-8}


\bibitem[Cheon et~al\mbox{.}(2020)]%
        {cheon2020efficient}
\bibfield{author}{\bibinfo{person}{Jung~Hee Cheon}, \bibinfo{person}{Dongwoo Kim}, {and} \bibinfo{person}{Duhyeong Kim}.} \bibinfo{year}{2020}\natexlab{}.
\newblock \showarticletitle{Efficient Homomorphic Comparison Methods with Optimal Complexity}. In \bibinfo{booktitle}{\emph{Advances in Cryptology -- ASIACRYPT 2020}}, \bibfield{editor}{\bibinfo{person}{Shiho Moriai} {and} \bibinfo{person}{Huaxiong Wang}} (Eds.). \bibinfo{publisher}{Springer International Publishing}, \bibinfo{address}{Cham}, \bibinfo{pages}{221--256}.
\newblock
\showISBNx{978-3-030-64834-3}


\bibitem[Cheon et~al\mbox{.}(2019)]%
        {cheon2019numerical}
\bibfield{author}{\bibinfo{person}{Jung~Hee Cheon}, \bibinfo{person}{Dongwoo Kim}, \bibinfo{person}{Duhyeong Kim}, \bibinfo{person}{Hun~Hee Lee}, {and} \bibinfo{person}{Keewoo Lee}.} \bibinfo{year}{2019}\natexlab{}.
\newblock \showarticletitle{Numerical Method for Comparison on Homomorphically Encrypted Numbers}. In \bibinfo{booktitle}{\emph{Advances in Cryptology -- ASIACRYPT 2019}}, \bibfield{editor}{\bibinfo{person}{Steven~D. Galbraith} {and} \bibinfo{person}{Shiho Moriai}} (Eds.). \bibinfo{publisher}{Springer International Publishing}, \bibinfo{address}{Cham}, \bibinfo{pages}{415--445}.
\newblock
\showISBNx{978-3-030-34621-8}


\bibitem[Chowdhury et~al\mbox{.}(2011)]%
        {Chowdhury2011}
\bibfield{author}{\bibinfo{person}{M.~F.~R. Chowdhury}, \bibinfo{person}{S.-A. Selouani}, {and} \bibinfo{person}{D. O'Shaughnessy}.} \bibinfo{year}{2011}\natexlab{}.
\newblock \showarticletitle{Bayesian on-line spectral change point detection: a soft computing approach for on-line {ASR}}.
\newblock \bibinfo{journal}{\emph{International Journal of Speech Technology}} \bibinfo{volume}{15}, \bibinfo{number}{1} (\bibinfo{year}{2011}), \bibinfo{pages}{5--23}.
\newblock
\href{https://doi.org/10.1007/s10772-011-9116-2}{doi:\nolinkurl{10.1007/s10772-011-9116-2}}


\bibitem[Cs{\"o}rg\H{o} and Horv{\'a}th(1997)]%
        {CsorgoHorvath1997}
\bibfield{author}{\bibinfo{person}{Mikl{\'o}s Cs{\"o}rg\H{o}} {and} \bibinfo{person}{Lajos Horv{\'a}th}.} \bibinfo{year}{1997}\natexlab{}.
\newblock \bibinfo{booktitle}{\emph{Limit Theorems in Change-Point Analysis}}.
\newblock \bibinfo{publisher}{John Wiley \& Sons}, \bibinfo{address}{Chichester}.
\newblock


\bibitem[Cummings et~al\mbox{.}(2018)]%
        {NEURIPS2018_f19ec2b8}
\bibfield{author}{\bibinfo{person}{Rachel Cummings}, \bibinfo{person}{Sara Krehbiel}, \bibinfo{person}{Yajun Mei}, \bibinfo{person}{Rui Tuo}, {and} \bibinfo{person}{Wanrong Zhang}.} \bibinfo{year}{2018}\natexlab{}.
\newblock \showarticletitle{Differentially Private Change-Point Detection}. In \bibinfo{booktitle}{\emph{Advances in Neural Information Processing Systems}}, \bibfield{editor}{\bibinfo{person}{S.~Bengio}, \bibinfo{person}{H.~Wallach}, \bibinfo{person}{H.~Larochelle}, \bibinfo{person}{K.~Grauman}, \bibinfo{person}{N.~Cesa-Bianchi}, {and} \bibinfo{person}{R.~Garnett}} (Eds.), Vol.~\bibinfo{volume}{31}. \bibinfo{publisher}{Curran Associates, Inc.}
\newblock
\urldef\tempurl%
\url{https://proceedings.neurips.cc/paper_files/paper/2018/file/f19ec2b84181033bf4753a5a51d5d608-Paper.pdf}
\showURL{%
\tempurl}


\bibitem[Davies et~al\mbox{.}(2012)]%
        {DAVIES20123623}
\bibfield{author}{\bibinfo{person}{Laurie Davies}, \bibinfo{person}{Christian Höhenrieder}, {and} \bibinfo{person}{Walter Krämer}.} \bibinfo{year}{2012}\natexlab{}.
\newblock \showarticletitle{Recursive computation of piecewise constant volatilities}.
\newblock \bibinfo{journal}{\emph{Computational Statistics \& Data Analysis}} \bibinfo{volume}{56}, \bibinfo{number}{11} (\bibinfo{year}{2012}), \bibinfo{pages}{3623--3631}.
\newblock
\showISSN{0167-9473}
\href{https://doi.org/10.1016/j.csda.2010.06.027}{doi:\nolinkurl{10.1016/j.csda.2010.06.027}}
\newblock
\shownote{1st issue of the Annals of Computational and Financial Econometrics Sixth Special Issue on Computational Econometrics}.


\bibitem[Dwork et~al\mbox{.}(2006)]%
        {dwork2006calibrating}
\bibfield{author}{\bibinfo{person}{Cynthia Dwork}, \bibinfo{person}{Frank McSherry}, \bibinfo{person}{Kobbi Nissim}, {and} \bibinfo{person}{Adam Smith}.} \bibinfo{year}{2006}\natexlab{}.
\newblock \showarticletitle{Calibrating noise to sensitivity in private data analysis}. In \bibinfo{booktitle}{\emph{Theory of cryptography conference}}. Springer, \bibinfo{pages}{265--284}.
\newblock


\bibitem[Dwork and Roth(2014)]%
        {dwork2014algorithmic}
\bibfield{author}{\bibinfo{person}{Cynthia Dwork} {and} \bibinfo{person}{Aaron Roth}.} \bibinfo{year}{2014}\natexlab{}.
\newblock \showarticletitle{The Algorithmic Foundations of Differential Privacy}.
\newblock \bibinfo{journal}{\emph{Found. Trends Theor. Comput. Sci.}} \bibinfo{volume}{9}, \bibinfo{number}{3–4} (\bibinfo{date}{Aug.} \bibinfo{year}{2014}), \bibinfo{pages}{211–407}.
\newblock
\showISSN{1551-305X}
\href{https://doi.org/10.1561/0400000042}{doi:\nolinkurl{10.1561/0400000042}}


\bibitem[Gupta et~al\mbox{.}(2024)]%
        {GUPTA2024123342}
\bibfield{author}{\bibinfo{person}{Muktesh Gupta}, \bibinfo{person}{Rajesh Wadhvani}, {and} \bibinfo{person}{Akhtar Rasool}.} \bibinfo{year}{2024}\natexlab{}.
\newblock \showarticletitle{Comprehensive analysis of change-point dynamics detection in time series data: A review}.
\newblock \bibinfo{journal}{\emph{Expert Systems with Applications}}  \bibinfo{volume}{248} (\bibinfo{year}{2024}), \bibinfo{pages}{123342}.
\newblock
\showISSN{0957-4174}
\href{https://doi.org/10.1016/j.eswa.2024.123342}{doi:\nolinkurl{10.1016/j.eswa.2024.123342}}


\bibitem[Jandhyala et~al\mbox{.}(2013)]%
        {jandhyala2013inference}
\bibfield{author}{\bibinfo{person}{Venkata Jandhyala}, \bibinfo{person}{Stergios Fotopoulos}, \bibinfo{person}{Ian MacNeill}, {and} \bibinfo{person}{Pengyu Liu}.} \bibinfo{year}{2013}\natexlab{}.
\newblock \showarticletitle{Inference for single and multiple change-points in time series}.
\newblock \bibinfo{journal}{\emph{Journal of Time Series Analysis}} \bibinfo{volume}{34}, \bibinfo{number}{4} (\bibinfo{year}{2013}), \bibinfo{pages}{423--446}.
\newblock
\showeprint{https://onlinelibrary.wiley.com/doi/pdf/10.1111/jtsa.12035}
\href{https://doi.org/10.1111/jtsa.12035}{doi:\nolinkurl{10.1111/jtsa.12035}}


\bibitem[Koumar et~al\mbox{.}(2025)]%
        {koumar2025cesnet}
\bibfield{author}{\bibinfo{person}{Jan Koumar}, \bibinfo{person}{Lukas Hynek}, \bibinfo{person}{Tom{\'a}{\v{s}} {\v{C}}ejka}, {et~al\mbox{.}}} \bibinfo{year}{2025}\natexlab{}.
\newblock \showarticletitle{CESNET-TimeSeries24: Time Series Dataset for Network Traffic Anomaly Detection and Forecasting}.
\newblock \bibinfo{journal}{\emph{Scientific Data}}  \bibinfo{volume}{12} (\bibinfo{year}{2025}), \bibinfo{pages}{338}.
\newblock
\href{https://doi.org/10.1038/s41597-025-04603-x}{doi:\nolinkurl{10.1038/s41597-025-04603-x}}


\bibitem[Lavielle and Teyssi{\`e}re(2007)]%
        {lavielle2007adaptive}
\bibfield{author}{\bibinfo{person}{Marc Lavielle} {and} \bibinfo{person}{Gilles Teyssi{\`e}re}.} \bibinfo{year}{2007}\natexlab{}.
\newblock \bibinfo{booktitle}{\emph{Adaptive Detection of Multiple Change-Points in Asset Price Volatility}}.
\newblock \bibinfo{publisher}{Springer Berlin Heidelberg}, \bibinfo{address}{Berlin, Heidelberg}, \bibinfo{pages}{129--156}.
\newblock
\showISBNx{978-3-540-34625-8}
\href{https://doi.org/10.1007/978-3-540-34625-8_5}{doi:\nolinkurl{10.1007/978-3-540-34625-8_5}}


\bibitem[L{\'e}vy-Leduc and Roueff(2009)]%
        {levy2009detection}
\bibfield{author}{\bibinfo{person}{C. L{\'e}vy-Leduc} {and} \bibinfo{person}{F. Roueff}.} \bibinfo{year}{2009}\natexlab{}.
\newblock \showarticletitle{Detection and localization of change-points in high-dimensional network traffic data}.
\newblock \bibinfo{journal}{\emph{The Annals of Applied Statistics}} \bibinfo{volume}{3}, \bibinfo{number}{2} (\bibinfo{year}{2009}), \bibinfo{pages}{637--662}.
\newblock
\href{https://doi.org/10.1214/08-AOAS232}{doi:\nolinkurl{10.1214/08-AOAS232}}


\bibitem[Lung-Yut-Fong et~al\mbox{.}(2012)]%
        {lung2012distributed}
\bibfield{author}{\bibinfo{person}{A. Lung-Yut-Fong}, \bibinfo{person}{C. L{\'e}vy-Leduc}, {and} \bibinfo{person}{O. Capp{\'e}}.} \bibinfo{year}{2012}\natexlab{}.
\newblock \showarticletitle{Distributed detection/localization of change-points in high-dimensional network traffic data}.
\newblock \bibinfo{journal}{\emph{Statistics and Computing}} \bibinfo{volume}{22}, \bibinfo{number}{2} (\bibinfo{year}{2012}), \bibinfo{pages}{485--496}.
\newblock
\href{https://doi.org/10.1007/s11222-011-9240-5}{doi:\nolinkurl{10.1007/s11222-011-9240-5}}


\bibitem[Mazzone et~al\mbox{.}(2025a)]%
        {mazzone2025privacy}
\bibfield{author}{\bibinfo{person}{Federico Mazzone}, \bibinfo{person}{Trevor Brown}, \bibinfo{person}{Florian Kerschbaum}, \bibinfo{person}{Kevin~H Wilson}, \bibinfo{person}{Maarten Everts}, \bibinfo{person}{Florian Hahn}, {and} \bibinfo{person}{Andreas Peter}.} \bibinfo{year}{2025}\natexlab{a}.
\newblock \showarticletitle{Privacy-Preserving Vertical K-Means Clustering}.
\newblock \bibinfo{journal}{\emph{arXiv preprint arXiv:2504.07578}} (\bibinfo{year}{2025}).
\newblock


\bibitem[Mazzone et~al\mbox{.}(2025b)]%
        {mazzone2025efficient}
\bibfield{author}{\bibinfo{person}{Federico Mazzone}, \bibinfo{person}{Maarten Everts}, \bibinfo{person}{Florian Hahn}, {and} \bibinfo{person}{Andreas Peter}.} \bibinfo{year}{2025}\natexlab{b}.
\newblock \showarticletitle{Efficient ranking, order statistics, and sorting under CKKS}. In \bibinfo{booktitle}{\emph{Proceedings of the 34th USENIX Security Symposium (USENIX Security 2025)}}.
\newblock
\urldef\tempurl%
\url{https://www.usenix.org/system/files/usenixsecurity25-mazzone.pdf}
\showURL{%
\tempurl}


\bibitem[Mironov(2017)]%
        {Mironov2017RnyiDP}
\bibfield{author}{\bibinfo{person}{Ilya Mironov}.} \bibinfo{year}{2017}\natexlab{}.
\newblock \showarticletitle{R{\'e}nyi Differential Privacy}.
\newblock \bibinfo{journal}{\emph{2017 IEEE 30th Computer Security Foundations Symposium (CSF)}} (\bibinfo{year}{2017}), \bibinfo{pages}{263--275}.
\newblock
\urldef\tempurl%
\url{https://api.semanticscholar.org/CorpusID:9386213}
\showURL{%
\tempurl}


\bibitem[Page(1955)]%
        {page1955test}
\bibfield{author}{\bibinfo{person}{E.~S. Page}.} \bibinfo{year}{1955}\natexlab{}.
\newblock \showarticletitle{A test for a change in a parameter occurring at an unknown point}.
\newblock \bibinfo{journal}{\emph{Biometrika}}  \bibinfo{volume}{42} (\bibinfo{year}{1955}), \bibinfo{pages}{523--527}.
\newblock
\href{https://doi.org/10.2307/2333401}{doi:\nolinkurl{10.2307/2333401}}


\bibitem[Peng et~al\mbox{.}(1999)]%
        {peng1999meditation}
\bibfield{author}{\bibinfo{person}{C.-K Peng}, \bibinfo{person}{Joseph~E Mietus}, \bibinfo{person}{Yanhui Liu}, \bibinfo{person}{Gurucharan Khalsa}, \bibinfo{person}{Pamela~S Douglas}, \bibinfo{person}{Herbert Benson}, {and} \bibinfo{person}{Ary~L Goldberger}.} \bibinfo{year}{1999}\natexlab{}.
\newblock \showarticletitle{Exaggerated heart rate oscillations during two meditation techniques}.
\newblock \bibinfo{journal}{\emph{International Journal of Cardiology}} \bibinfo{volume}{70}, \bibinfo{number}{2} (\bibinfo{year}{1999}), \bibinfo{pages}{101--107}.
\newblock
\href{https://doi.org/10.1016/S0167-5273(99)00066-2}{doi:\nolinkurl{10.1016/S0167-5273(99)00066-2}}


\bibitem[Reeves et~al\mbox{.}(2007)]%
        {reeves2007review}
\bibfield{author}{\bibinfo{person}{J. Reeves}, \bibinfo{person}{J. Chen}, \bibinfo{person}{X.~L. Wang}, \bibinfo{person}{R. Lund}, {and} \bibinfo{person}{Q.~Q. Lu}.} \bibinfo{year}{2007}\natexlab{}.
\newblock \showarticletitle{A review and comparison of changepoint detection techniques for climate data}.
\newblock \bibinfo{journal}{\emph{Journal of Applied Meteorology and Climatology}} \bibinfo{volume}{46}, \bibinfo{number}{6} (\bibinfo{year}{2007}), \bibinfo{pages}{900--915}.
\newblock
\showISSN{1558-8424}
\href{https://doi.org/10.1175/JAM2493.1}{doi:\nolinkurl{10.1175/JAM2493.1}}


\bibitem[Schmidt et~al\mbox{.}(2012)]%
        {SCHMIDT20129}
\bibfield{author}{\bibinfo{person}{Robert Schmidt}, \bibinfo{person}{Cornelius Krasselt}, {and} \bibinfo{person}{Christian {von Borczyskowski}}.} \bibinfo{year}{2012}\natexlab{}.
\newblock \showarticletitle{Change point analysis of matrix dependent photoluminescence intermittency of single CdSe/ZnS quantum dots with intermediate intensity levels}.
\newblock \bibinfo{journal}{\emph{Chemical Physics}}  \bibinfo{volume}{406} (\bibinfo{year}{2012}), \bibinfo{pages}{9--14}.
\newblock
\showISSN{0301-0104}
\href{https://doi.org/10.1016/j.chemphys.2012.02.018}{doi:\nolinkurl{10.1016/j.chemphys.2012.02.018}}
\newblock
\shownote{Single molecule spectroscopy: Current status and perspectives}.


\bibitem[Staudacher et~al\mbox{.}(2005)]%
        {Staudacher2005}
\bibfield{author}{\bibinfo{person}{M. Staudacher}, \bibinfo{person}{S. Telser}, \bibinfo{person}{A. Amann}, \bibinfo{person}{H. Hinterhuber}, {and} \bibinfo{person}{M. Ritsch-Marte}.} \bibinfo{year}{2005}\natexlab{}.
\newblock \showarticletitle{A new method for change-point detection developed for on-line analysis of the heart beat variability during sleep}.
\newblock \bibinfo{journal}{\emph{Physica A: Statistical Mechanics and its Applications}} \bibinfo{volume}{349}, \bibinfo{number}{3} (\bibinfo{year}{2005}), \bibinfo{pages}{582--596}.
\newblock
\showISSN{0378-4371}
\href{https://doi.org/10.1016/j.physa.2004.10.026}{doi:\nolinkurl{10.1016/j.physa.2004.10.026}}


\bibitem[Terzano et~al\mbox{.}(2001)]%
        {terzano2001cap}
\bibfield{author}{\bibinfo{person}{M.~G. Terzano}, \bibinfo{person}{L. Parrino}, \bibinfo{person}{A. Sherieri}, \bibinfo{person}{R. Chervin}, \bibinfo{person}{S. Chokroverty}, \bibinfo{person}{C. Guilleminault}, \bibinfo{person}{M. Hirshkowitz}, \bibinfo{person}{M. Mahowald}, \bibinfo{person}{H. Moldofsky}, \bibinfo{person}{A. Rosa}, \bibinfo{person}{R. Thomas}, {and} \bibinfo{person}{A. Walters}.} \bibinfo{year}{2001}\natexlab{}.
\newblock \showarticletitle{Atlas, rules, and recording techniques for the scoring of cyclic alternating pattern (CAP) in human sleep}.
\newblock \bibinfo{journal}{\emph{Sleep Medicine}} \bibinfo{volume}{2}, \bibinfo{number}{6} (\bibinfo{date}{Nov.} \bibinfo{year}{2001}), \bibinfo{pages}{537--553}.
\newblock
\href{https://doi.org/10.1016/S1389-9457(01)00149-6}{doi:\nolinkurl{10.1016/S1389-9457(01)00149-6}}
\newblock
\shownote{Erratum in: Sleep Med. 2002 Mar;3(2):185. PMID: 14592270}.


\bibitem[Verbesselt et~al\mbox{.}(2010)]%
        {verbesselt2010detecting}
\bibfield{author}{\bibinfo{person}{Jan Verbesselt}, \bibinfo{person}{Rob Hyndman}, \bibinfo{person}{Glenn Newnham}, {and} \bibinfo{person}{Darius Culvenor}.} \bibinfo{year}{2010}\natexlab{}.
\newblock \showarticletitle{Detecting trend and seasonal changes in satellite image time series}.
\newblock \bibinfo{journal}{\emph{Remote Sensing of Environment}} \bibinfo{volume}{114}, \bibinfo{number}{1} (\bibinfo{year}{2010}), \bibinfo{pages}{106--115}.
\newblock
\showISSN{0034-4257}
\href{https://doi.org/10.1016/j.rse.2009.08.014}{doi:\nolinkurl{10.1016/j.rse.2009.08.014}}


\bibitem[Zhang et~al\mbox{.}(2021)]%
        {JMLR:v22:19-770}
\bibfield{author}{\bibinfo{person}{Wanrong Zhang}, \bibinfo{person}{Sara Krehbiel}, \bibinfo{person}{Rui Tuo}, \bibinfo{person}{Yajun Mei}, {and} \bibinfo{person}{Rachel Cummings}.} \bibinfo{year}{2021}\natexlab{}.
\newblock \showarticletitle{Single and Multiple Change-Point Detection with Differential Privacy}.
\newblock \bibinfo{journal}{\emph{Journal of Machine Learning Research}} \bibinfo{volume}{22}, \bibinfo{number}{29} (\bibinfo{year}{2021}), \bibinfo{pages}{1--36}.
\newblock
\urldef\tempurl%
\url{http://jmlr.org/papers/v22/19-770.html}
\showURL{%
\tempurl}


\end{thebibliography}

\appendix

\section{Encrypted Matrix Operations}
\label{app:enc-matrix-operations}

Here we provide some well-known algorithms for performing operations on encrypted square matrices of size $\vectorLength$, which must be a power of two.

\begin{algorithm}
\caption{$\sumR$}
\label{alg:sumR}
\begin{algorithmic}[1]
\Require $X$ encryption of a square matrix of size $\vectorLength$.
\Ensure $X$ encryption of a row vector.
\For{$i = 0, \dots, \log{\vectorLength} - 1$}
    \State $X \gets X + (X \ll \vectorLength \cdot 2^i)$
\EndFor
\State $X \gets X \cdot (1^\vectorLength \parallel 0^{\vectorLength(\vectorLength - 1)})$
\State \Return $X$
\end{algorithmic}
\end{algorithm}

\begin{algorithm}
\caption{$\sumC$}
\label{alg:sumC}
\begin{algorithmic}[1]
\Require $X$ encryption of a square matrix of size $\vectorLength$.
\Ensure $X$ encryption of a column vector.
\For{$i = 0, \dots, \log{\vectorLength} - 1$}
    \State $X \gets X + (X \ll 2^i)$
\EndFor
\State $X \gets X \cdot (1 \parallel 0^{\vectorLength - 1})^\vectorLength$
\State \Return $X$
\end{algorithmic}
\end{algorithm}

\begin{algorithm}
\caption{$\replR$}
\label{alg:replR}
\begin{algorithmic}[1]
\Require $X$ encryption of a row vector of size $\vectorLength$.
\Ensure $X$ encryption of a square matrix.
\For{$i = 0, \dots, \log{\vectorLength} - 1$}
    \State $X \gets X + (X \gg \vectorLength \cdot 2^i)$
\EndFor
\State \Return $X$
\end{algorithmic}
\end{algorithm}

\begin{algorithm}
\caption{$\replC$}
\label{alg:replC}
\begin{algorithmic}[1]
\Require $X$ encryption of a column vector of size $\vectorLength$.
\Ensure $X$ encryption of a square matrix.
\For{$i = 0, \dots, \log{\vectorLength} - 1$}
    \State $X \gets X + (X \gg 2^i)$
\EndFor
\State \Return $X$
\end{algorithmic}
\end{algorithm}

\begin{algorithm}
\caption{$\transR$}
\label{alg:transR}
\begin{algorithmic}[1]

\Require $X$ encryption of a vector $x$ encoded as a row.

\Ensure $X$ encryption of the vector $x$ encoded as a column.

\For{$i = 1, \dots, \ceil{\log{\vectorLength}}$}
    \State $X \gets X + (X \gg \vectorLength(\vectorLength-1) / 2^i)$
\EndFor
\State $X \gets X \cdot (1 \parallel 0^{\vectorLength - 1})^\vectorLength$
\State \Return $X$
\end{algorithmic}
\end{algorithm}

\begin{algorithm}
\caption{$\transC$}
\label{alg:transC}
\begin{algorithmic}[1]

\Require $X$ encryption of a vector $x$ encoded as a column.

\Ensure $X$ encryption of the vector $x$ encoded as a row.

\For{$i = 1, \dots, \ceil{\log{\vectorLength}}$}
    \State $X \gets X + (X \ll \vectorLength(\vectorLength-1) / 2^i)$
\EndFor
\State $X \gets X \cdot (1^\vectorLength \parallel 0^{\vectorLength(\vectorLength - 1)})$
\State \Return $X$
\end{algorithmic}
\end{algorithm}

\section{Differential Privacy Background}  
\label{appendix:DP}
We first report the standard definition of \((\varepsilon, \delta)\)-Differential Privacy, and Rényi Differential Privacy (RDP), which generalizes standard DP and is particularly useful for analyzing the privacy properties of mechanisms involving the composition of multiple queries.  For a more comprehensive discussion on these concepts, we refer to Mironov (2017) \cite{Mironov2017RnyiDP} and Dwork and Roth (2014) \cite{dwork2014algorithmic}.  
\begin{definition}
\label{def:DP} Let $\varepsilon>0$ and $\delta>0$. 
A randomized mechanism \( \mathcal{A}: \mathcal{Z}^n \to \mathcal{R} \) satisfies \((\varepsilon, \delta)\)-Differential Privacy (DP) if, for all adjacent datasets \( S_1, S_2 \in \mathcal{Z}^n \) (i.e., differing by at most one element) and for all measurable subsets \( \mathcal{S} \subseteq \mathcal{R} \), we have  
\[
\mathbb{P}[\mathcal{A}(S_1) \in \mathcal{S}] \leq e^\varepsilon \mathbb{P}[\mathcal{A}(S_2) \in \mathcal{S}] + \delta.
\]
\end{definition}  
The parameter \( \varepsilon \) controls the privacy loss, with smaller values ensuring stronger privacy guarantees, while \( \delta \) accounts for the probability of violating pure \(\varepsilon\)-DP.  

Next, the Rényi Differential Privacy is a relaxation of \((\varepsilon, \delta)\)-DP based on Rényi divergence, which facilitates tighter privacy accounting, especially under composition.  
\begin{definition}
For two probability distributions \( f \) and \( g \) supported over \( \mathcal{R} \), the Rényi divergence of order \( \beta > 1 \) is defined as  
\[
D_\beta(f\| g) :=\frac{1}{\beta-1} \log \mathbb{E}_{x \sim g}\left[\left(\frac{f(x)}{g(x)}\right)^\beta\right].  
\]
\end{definition}  

\begin{definition}
\((\beta, \varepsilon)\)-Rényi Differential Privacy (RDP). A randomized mechanism \( \mathcal{A}: \mathcal{Z}^n \to \mathcal{R} \) is said to be \(\varepsilon\)-Rényi differentially private of order \( \beta \) (denoted as \((\beta, \varepsilon)\)-RDP) if  
\begin{equation}
\label{def:reny-dp}
    \sup_ { S_1 \sim S_2 } D_\beta\left(\mathcal{A}(S_1) \| \mathcal{A}(S_2)\right) \leq \varepsilon.
\end{equation}
\end{definition}  

\textbf{Properties of Rényi Differential Privacy } 

We conclude this section by summarizing three fundamental properties of RDP \cite{Mironov2017RnyiDP}:  

\begin{enumerate}
    \item \textbf{Post-processing}: If \( \mathcal{A} \) is \((\beta, \varepsilon)\)-RDP and \( g \) is a randomized mapping, then \( g \circ \mathcal{A} \) is also \((\beta, \varepsilon)\)-RDP.
    
    \item \textbf{Adaptive Composition}: If \( \mathcal{A}_1 : \mathcal{Z}^n \to \mathcal{X}_1 \) is \((\beta, \varepsilon_1)\)-RDP, and \( \mathcal{A}_2 : \mathcal{X}_1 \times \mathcal{Z}^n \to \mathcal{X}_2 \) is \((\beta, \varepsilon_2)\)-RDP, then the combined mechanism  
    \[
    (\mathcal{A}_1(\cdot), \mathcal{A}_2(\mathcal{A}_1(\cdot), \cdot))
    \]
    satisfies \((\beta, \varepsilon_1 + \varepsilon_2)\)-RDP.  
    
    \item \textbf{From RDP to DP}: If \( \mathcal{A} \) is \((\beta, \varepsilon)\)-RDP, then it satisfies  
    \[
    \left(\varepsilon + \frac{\log(1/\delta)}{\beta - 1}, \delta\right)\text{-DP}
    \]
    for any \( 0 < \delta < 1 \).
\end{enumerate}  

\subsection{Releasing Time Series}

\textbf{Assumption}
Let $X_t$ admit a probability density $f \in C^0_c(\mathbb{R}) \cap L^\infty(\mathbb{R})$, i.e.\ $f$ is continuous, bounded, and compactly supported.  

\noindent
As a consequence, $X_t$ is almost surely bounded, meaning there exists $M>0$ such that 
\[
\mathbb{P}(|X_t| < M) = 1.
\]
Hence, the support of the data vector $\mathbb{X} := (X_1,\ldots,X_n)$ is given by
\[
\mathcal{Z}^n = [-M,M]^n.
\]
Note that Assumption 1 is an analogous of the truncation condition imposed by Zhang et al (Section 2.1 in 
 \cite{JMLR:v22:19-770}). 
\begin{definition}[Additive Gaussian mechanism]
\label{def:mechanism}
For $\sigma > 0$, let $(\varepsilon_j)_{j=1}^n$ be i.i.d.\ $\mathcal{N}(0,\sigma^2)$ random variables.  
We define the mechanism $\mathcal{A} : \mathcal{Z}^n \to \mathbb{R}^n$ by
\begin{align}
\label{eq:mechanism}
    \mathcal{A}(\mathbb{X}) 
    := \mathbb{X} + \boldsymbol{\varepsilon} 
\end{align}
where $\boldsymbol{\varepsilon} = (\varepsilon_1,\ldots,\varepsilon_n)$ and $\mathbf{1}=(1,\ldots,1)^\top$.  
In particular, each coordinate satisfies
\[
(\mathcal{A}(\mathbb{X}))_i \sim \mathcal{N}(X_i ,\, \sigma^2),
\]
independently across $i$.
\end{definition}

\noindent

\begin{proposition}[Gaussian RDP, equal variances]
\label{prop:gauss-rdp}
Let $\beta>1$ and $P = \mathcal{N}(\mu_1,\sigma^2)$, $Q = \mathcal{N}(\mu_2,\sigma^2)$. Then
\[
D_\beta(P \,\|\, Q) = \frac{\beta}{2\sigma^2} (\mu_1-\mu_2)^2.
\]
\end{proposition}

\begin{corollary}[RDP of $\mathcal{A}$]
\label{cor:rdp}
For neighboring datasets $\mathbb{X}, \mathbb{X}'$ differing only at index $k$, the mechanism \eqref{eq:mechanism} satisfies
\[
D_\beta(\mathcal{A}(\mathbb{X}) \,\|\, \mathcal{A}(\mathbb{X}')) 
= \frac{\beta}{2\sigma^2} (X_k - X_k')^2.
\]
Since $|X_k - X_k'|\le 2M$, we obtain the uniform bound
\[
D_\beta(\mathcal{A}(\mathbb{X}) \,\|\, \mathcal{A}(\mathbb{X}')) 
\;\le\; \frac{2\beta M^2}{\sigma^2}.
\]
\end{corollary}
\begin{proof}
    For two neighboring datasets \( \mathbb{X}, \mathbb{X}' \) that differ only in the \( k \)-th entry, it follows that  
\begin{align*}
    &D_\beta (\mathcal{A}(\mathbb{X}) \| \mathcal{A}(\mathbb{X}')) \\
    &= \frac{1}{\beta-1} \log \mathbb{E}_{x_i\sim \mathcal{N}( X_i ,\sigma^2 )} \left( \frac{\prod_i \mathcal{N}(X_i , \sigma^2)}{\prod_i \mathcal{N}(X'_i , \sigma^2 )} \right) \\
    &= \frac{1}{\beta -1} \log \mathbb{E}_{x_k\sim \mathcal{N}( X_k ,\sigma^2 )} \left( \frac{\mathcal{N}(X_k , \sigma^2 )}{ \mathcal{N}(X'_k , \sigma^2 )} \right) \\
    &=  D_\beta \left( (\mathcal{A}(\mathbb{X}))_k \| (\mathcal{A}(\mathbb{X}'))_k \right).
\end{align*}  
\end{proof}
\begin{theorem}[$(\varepsilon,\delta)$-DP guarantee]
\label{thm:dp}
For any $\delta \in (0,1)$, the mechanism $\mathcal{A}$ defined in \eqref{eq:mechanism} is $(\varepsilon,\delta)$-differentially private with
\begin{align}
\varepsilon
&= \frac{\Delta^2}{2\sigma^2} + \frac{\sqrt{2}\,\Delta}{\sigma}\,\sqrt{\log\!\left(\tfrac{1}{\delta}\right)},
\end{align}
where $\Delta := \sup_{\mathbb{X}\sim \mathbb{X}'} |X_k-X_k'| \le 2M$ denotes the global $\ell_2$-sensitivity.  
In particular,
\[
\varepsilon \;\le\; \frac{2M^2}{\sigma^2} + \frac{2\sqrt{2}\,M}{\sigma}\,\sqrt{\log\!\left(\tfrac{1}{\delta}\right)}.
\]
\end{theorem}

\begin{proof}
From Corollary~\ref{cor:rdp}, the mechanism is $(\beta,\varepsilon_{\mathrm{RDP}}(\beta))$-RDP with
\[
\varepsilon_{\mathrm{RDP}}(\beta) = \frac{\beta\Delta^2}{2\sigma^2}.
\]
The standard RDP-to-DP conversion gives, for any $\beta>1$,
\[
\varepsilon = \varepsilon_{\mathrm{RDP}}(\beta) + \frac{\log(1/\delta)}{\beta-1}.
\]
Optimizing in $\beta$ yields $\beta^\star = 1 + \sqrt{\tfrac{2\sigma^2 \log(1/\delta)}{\Delta^2}}$, which implies the stated bound.
\end{proof}
\begin{corollary}
    \label{cor:local_DP}
    Let $(X_t)_{t=1, \ldots, n}$ be a time series satisfying Assumption 1. For any $\varepsilon>0$ and any $\delta \in (0,1)$, the mechanism $\mathcal{A}$ defined in \eqref{eq:mechanism} with $\sigma_\text{DP}= \frac{\sqrt{2}\,M}{\varepsilon}\big(\sqrt{\log(1/\delta)}+\sqrt{\log(1/\delta)+\varepsilon}\big)$ is $(\varepsilon,\delta)$-differentially private .
\end{corollary}
\begin{proof}
Let \(L:=\log(1/\delta)\). Solving \(\varepsilon=\Delta^2/(2\sigma^2)+(\sqrt{2}\,\Delta/\sigma)\sqrt{L}\) for \(\sigma\) gives
\[
\sigma=\frac{\Delta}{\sqrt{2}\,\varepsilon}\big(\sqrt{L}+\sqrt{L+\varepsilon}\big)
\]
With \(\Delta=2M\) yields $\sigma_\text{DP}$ as stated in the Corollary. 
\end{proof}

\subsection{Local DP-CPD}
Corollary~\ref{cor:local_DP} specifies the amount of noise required for the mechanism \(\mathcal{A}\) to satisfy \((\varepsilon, \delta)\)-differential privacy. This guarantee holds under Assumption~1, which imposes bounded sensitivity on the input time series. When this assumption is not fulfilled—i.e., when the time series may contain unbounded or heavy-tailed values—it is sufficient to apply a componentwise Huber-type clipping transformation:
\[
\phi_M(X_t) = 
\begin{cases}
X_t, & \text{if } |X_t| \leq M, \\
M,   & \text{if } X_t > M, \\
-M,  & \text{if } X_t < -M,
\end{cases}
\]
for some truncation level \(M > 0\).

We then define the transformed process \(Y_t = \phi_M(X_t)\), and apply a CUSUM-type statistic to \(\{Y_t\}\). Since differential privacy is preserved under post-processing, the overall procedure remains \((\varepsilon, \delta)\)-DP. That is, if an algorithm is \((\varepsilon, \delta)\)-DP, then any deterministic function applied to its output also satisfies \((\varepsilon, \delta)\)-DP, see \cite{dwork2014algorithmic}.

\section{Central DP CPD simulations}
\label{app:central_DP}
This section reports results for Algorithm~3 of Zhang et al.~\cite{JMLR:v22:19-770} and Algorithm~1 of~\cite{NEURIPS2018_f19ec2b8}. We evaluate both methods on synthetic time series with a single change in the data-generating distribution. For a series of length \(n\) and a fixed fraction \(\tilde{\tau}\in(0,1)\), the true change-point is \(\tau=\lfloor \tilde{\tau} n\rfloor\), and the observations are independent with a distributional break at \(\tau\):
\begin{align*} &X_t\sim \mathcal{N}(\mu_1, \sigma_1^2)\quad\text{for }\,\, t=1, \ldots,\tau \\ &X_t\sim \text{Lap}(\mu_2, \sigma_2)\quad\text{for }\,\,t=\tau+1, \ldots,n \end{align*}
where the post-change Laplace law is parameterized so that \(\sigma_2\) denotes its standard deviation (equivalently, scale \(b=\sigma_2/\sqrt{2}\)). 

Following the protocol of Section~\ref{sec:experiments_comparisonDP}, for each setting we simulate \(N\) independent series \(\{X^{(i)}_t\}_{t=1}^n\), apply both methods, and record the estimated change \(\hat{k}^{(i)}\). Performance is summarized by the mean absolute relative error. 
All remaining experimental choices (privacy calibration, grids over \(\varepsilon\), etc.) match Section~\ref{sec:experiments_comparisonDP}.

We report the accuracy of both methods in the case of change of variance in Figure~\ref{fig:centralDPsetting1}, and in the case of change of both mean and variance in Figure~\ref{fig:centralDPsetting2}.

\begin{figure}[h!]
    \centering
      \resizebox{\linewidth}{!}{
    \begin{tikzpicture}
      \pgfplotsset{scaled y ticks=false}
      \begin{axis}[
        xlabel={Privacy Budget ($\varepsilon$)},
        ylabel={Relative Error},
        grid=both,
        major grid style={line width=.2pt, draw=gray!30},
        minor grid style={line width=.1pt, draw=gray!10},
        every axis plot/.append style={line width=1.1pt},
        mark options={solid},
        enlarge x limits=0.02,
        enlarge y limits=0.05,
        ymin=0,
        height=6cm, width=\linewidth,
        legend style={at={(0.5,1.05)}, anchor=south, draw=none, yshift=-2pt},
        legend columns=1, legend cell align={left}
      ]

       \pgfplotsset{cycle list={
            {myBlue, mark=triangle*, solid},
            {myGreen, mark=o, dashed},
            {myRed, mark=+, solid},
            {myPurple, mark=square*, dashed}
        }}
        \input{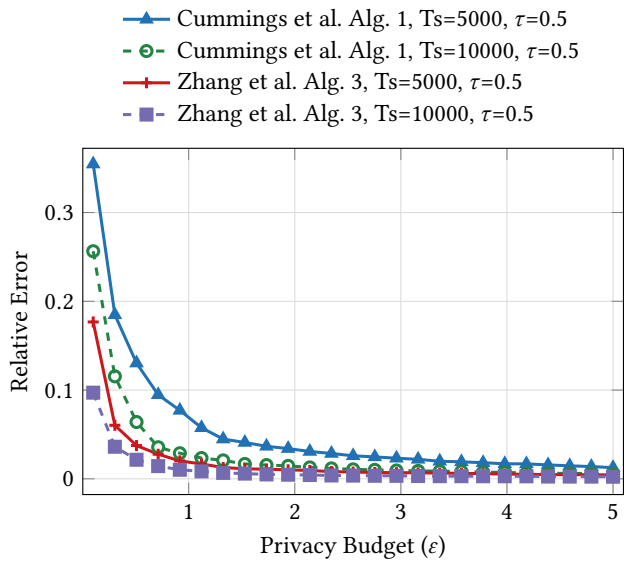}

      \end{axis}
    \end{tikzpicture}
    }
    \caption{Change of variance: $\sigma_1 \neq \sigma_2$ and $\mu_1 = \mu_2$.}
    \label{fig:centralDPsetting1}
\end{figure}

\begin{figure}[h!]
    \centering
    \resizebox{\linewidth}{!}{
    \begin{tikzpicture}
      \pgfplotsset{scaled y ticks=false}
      \begin{axis}[
        xlabel={Privacy Budget ($\varepsilon$)},
        ylabel={Relative Error},
        grid=both,
        major grid style={line width=.2pt, draw=gray!30},
        minor grid style={line width=.1pt, draw=gray!10},
        every axis plot/.append style={line width=1.1pt},
        mark options={solid},
        enlarge x limits=0.02,
        enlarge y limits=0.05,
        ymin=0,
        height=6cm, width=\linewidth,
        legend style={at={(0.5,1.05)}, anchor=south, draw=none, yshift=-2pt},
        legend columns=1, legend cell align={left}
      ]

        \pgfplotsset{cycle list={
            {myBlue, mark=triangle*, solid},
            {myGreen, mark=o, dashed},
            {myRed, mark=+, solid},
            {myPurple, mark=square*, dashed}
        }}

        \input{pics/setting2_tau_0.5_addplots.tex}

      \end{axis}
    \end{tikzpicture}
    }
    \caption{Change of mean and variance: $\sigma_1 \neq \sigma_2$ and $\mu_1\neq \mu_2$.}
    \label{fig:centralDPsetting2}
\end{figure}

\end{document}